 \PassOptionsToPackage{square,numbers,sort&compress}{natbib}
 \documentclass[final,1p,times,twocolumn]{elsarticle}

\usepackage{graphics}
\usepackage{graphicx}
\usepackage{epsfig}
\usepackage{hyperref}
\usepackage{amsthm,amssymb,amsmath,lineno,hyperref}
\theoremstyle{definition}
\newtheorem{definition}{Definition}[section]
\theoremstyle{theorem}
\newtheorem{theorem}[definition]{Theorem}
\newtheorem{lemma}[definition]{Lemma}
\newtheorem{proposition}[definition]{Proposition}
\newtheorem{corollary}[definition]{Corollary}
\newtheorem{remark}[definition]{Remark}
\newtheorem{Assumption}[definition]{Assumption}
\theoremstyle{definition}
\newtheorem{Example}[definition]{Example}

\usepackage{subfig}
\usepackage{float}
\usepackage{color}

\begin{document}

\begin{frontmatter}

\title{
{\bf \Large  Pricing and delta computation in jump-diffusion models with stochastic intensity by Malliavin calculus} \vspace{1cm} }
\author{Ayub Ahmadi}
\ead{Ayubahmadi66@gmail.com},
\author{Mahdieh Tahmasebi}
\ead{tahmasebi@modares.ac.ir}
\address{Department of Applied Mathematics, Tarbiat Modares University,\\
P.O. Box 14115-175, Tehran, IRAN.}
\begin{abstract}
{In this paper, the pricing of financial derivatives and the calculation of their delta Greek are investigated when the underlying asset is a jump-diffusion process in which the stochastic intensity component follows the CIR process. Utilizing Malliavin derivatives for pricing financial derivatives and challenging to find the Malliavin weight for accurately calculating delta will be established in such models. Because asset prices rely on information from the intensity process, the moments of the Malliavin weights and the underlying asset is crucial to be bounded. We apply the Euler scheme to show the convergence of the approximated solution, a financial derivative, and its delta Greeks, and we have established the convergence analysis. Our approach has been validated through numerical experiments, highlighting its effectiveness and potential for risk management and hedging strategies in markets characterized by jump and stochastic intensity dynamics.}
\end{abstract}

\end{frontmatter}
%

\section{Introduction}
Stochastic intensity is a critical attribute in financial markets, as it enhances the realism of event arrival rates within models. Unlike constant intensity, stochastic intensity allows for greater flexibility in capturing the randomness of event occurrences. This capability significantly improves the representation of new information arrivals, changes in investor behavior, and the occurrence of jumps such as market crashes, large price movements, or sudden shifts in volatility.
Financial institutions, portfolio managers, and investors can leverage these models to evaluate the likelihood and impact of extreme events, facilitating more informed decision-making and the development of effective risk mitigation strategies (see references \cite{Damiano} and \cite{leung}). Furthermore, the application of stochastic jump intensity to assess firms' default rates for risk evaluation and portfolio management is addressed in references \cite{feng} and \cite{levy}.
The self-exciting point process, where the current event intensity is influenced by past events, was first introduced in \cite{hawkes}. The critical role of jumps with stochastic intensity in option pricing is supported by the empirical findings in \cite{fang2000}, and its significance in modeling jump intensity risk is empirically validated in \cite{Santa2010}. 
In Markov intensity models with discrete states, known as Markov-modulated jump models, the pricing of risky underlying assets has been studied in \cite{elliot2007}, \cite{bo2010}, \cite{chang2013}, and more recently in \cite{shan2023}. In the latter, the Markov-modulated jump-diffusion process was utilized to model discrete dividend processes in financial markets. 
Within a continuous framework, \cite{Damiano} derived an analytical formula for pricing credit derivatives under Cox-Ingersoll-Ross (CIR) stochastic intensity models. Subsequently, \cite{brigo} extended this work by incorporating a smile-adjusted jump stochastic intensity to price credit default swaptions. Non-Gaussian intensity models were explored in \cite{bianchi}. 
In 2019, the authors in \cite{yang} proposed these models for variance exchange rates to price variance swaps. Additionally, studies in\cite{huang}, \cite{chang}, and more recently \cite{ma}, focused on option pricing under a double-exponential jump model with stochastic volatility and stochastic intensity, employing Fourier transform techniques.
On the other hand, Malliavin calculus is an advanced mathematical framework that extends traditional calculus to differentiate random variables and quantify their sensitivities. It plays a crucial role in accurately calculating the delta, pricing financial derivatives, designing hedging strategies, and informing investment decisions (see, for instance, \cite{Alos}, \cite{n6}, \cite{n7}, \cite{n8}, \cite{n12}). 
In 2004, the authors in \cite{el} computed the Greeks in a market driven by a discontinuous process with Poisson jump times and random jump sizes, utilizing Malliavin calculus on Poisson space.
Numerical simulations for the delta and gamma of Asian options demonstrate the efficiency of this approach compared to classical finite difference Monte Carlo approximations of derivatives. In \cite{huehne}, stochastic weights were derived for the fast and accurate computation of Greeks for options whose underlying asset is driven by a pure-jump Lévy process. Subsequently, Bavouzet and Messaoud explored this topic in \cite{bavouzet}, employing both the Malliavin derivative concerning jump amplitudes and the Wiener process. The computation of delta using Malliavin calculus for options on underlying assets modeled by Lévy processes is discussed in \cite{n19}, \cite{khedher}, \cite{matchie}, and \cite{n18}. For further details on Malliavin calculus applied to Lévy processes, readers are referred to \cite{proske} and \cite{n3}.
Recently, sensitivity analysis with respect to stock prices for singular stochastic differential equations (SDEs) was considered in \cite{n18}, and the regularity of distribution-dependent SDEs with jump processes was established in \cite{song} using Malliavin calculus. Additionally, \cite{hudde} presented a closed-form expression for Asian Greeks within an exponential Lévy process model.

In this article, we focus on jump-diffusion models with stochastic intensity, specifically the CIR model, also known as the self-exciting Cox process. Our investigation centers on the pricing of financial derivatives and deriving an expression for delta calculation using a Malliavin weight. By incorporating the Malliavin derivative of the intensity into the Malliavin derivative of the underlying process, we identify certain Wiener directions that belong to the domain of the Skorokhod operator in the Gaussian case. This result, presented in Theorem \ref{th1}, is utilized in the duality formula for calculating the delta and pricing financial derivatives.
Meanwhile, the use of conditional expectation to the information of the intensity is essential. It is important to note that there are two different approaches to defining the Malliavin derivative to jump processes (see chapters 10 and 11 of \cite{n3}). As a result, we will encounter two distinct Skorokhod integrals, each corresponding to the Malliavin weight associated with its respective approach in the calculation of delta.\\
Numerical analysis of the Euler scheme for pricing and delta computation in jump-diffusion models examines its convergence and error characteristics. Convergence for various discretization methods in jump-diffusion models has been studied in many papers, we mention some of them, here.
For instance, in \cite{glasserman2004convergence}, the weak convergence of the Euler scheme for simulating jump-diffusion processes with state-dependent jump intensities has been investigated.  In \cite{giesecke2018numerical} the authors demonstrated the convergence of the weak Euler method for jump-diffusion processes with drift, volatility, jump intensity, and jump size that depend on the general state. In papers \cite{bruti2007strong} and \cite{higham2005numerica}, strong convergence for stochastic differential equations driven by jumps with globally and locally Lipschitz coefficients have been proved.
In all these studies, the convergence of the discretized method has predominantly been investigated for models where jumps occur with constant intensity. In contrast, we demonstrate the convergence of the Euler method for models whose jump sizes are driven by stochastic intensities, as observed for computing the delta Greeks required in credit risk assessments.\\
This article is organized as follows: In Section \ref{sec3}, we review Malliavin's calculus on Wiener and Poisson spaces. In Section \ref{sec4}, we introduce the main model with a stochastic intensity process and examine the conditions necessary for the existence of the solution and the boundedness of its moments. We will derive the Malliavin derivative of the solution and identify directions where its derivative is invertible. In Section \ref{sec5}, we calculate the delta and price of the European option. We prove the convergence of the Euler method of asset price and convergence of the approximated delta Greeks in Section \ref{sec6}. Finally, in Section \ref{sec7}, we illustrate our main results and compare them with the finite difference method. In addition, the error of convergence of the discretized method will be established.
\section{A review on Malliavin calculus}\label{sec3}
Let us review some concepts of Malliavin calculus on Wiener space and in the Poisson framework, See standard reference \cite{n3}.  
\subsection{Malliavin calculus concepts on Wiener space}
For a positive real number $T$, suppose that 
$\Omega:=C_0(\left[0,T\right])$ is the space of real continuous functions $w$ on
$\left[0,T\right]$ with $w(0)=0$ equipped with the uniform norm
\begin{equation}
\rVert{w}\rVert_\infty=\sup_{t\in\left[0,T\right]}\lvert{w(t)}\rvert.
\end{equation}
Consider 
$\Big(\Omega,\mathcal{F}, \mathcal{F}_t, P\Big)$ as a filtered probability space, with coordinate map $t\to{W(t,w)}$ for Brownian motion $B(t)$ corresponding to the filtration
$\{\mathcal{F}_t\}$.
For every $\gamma\in\Omega$ of the Cameron-Martin space, the set of the functions in the form $\gamma(t)=\int_{0}^{t}g(s)ds$ for some $g\in{L^{2}}(\left[0,T\right])$,  and a random variable $F:\Omega\to{\mathbb{R}}$, the directional derivative of $F$ in the $\gamma$ direction, have defined as the following form, if the limit exists. In fact,
	\begin{equation*}
	D^W_\gamma{F}(w)=\frac{d}{d\epsilon}[F(w+\epsilon\gamma)]_{\epsilon=0}.
	\end{equation*}
If there exists some $\psi\in{L^{2}(\left[0,T\right]\times\Omega)}$ satisfying the following equation
	\begin{equation}
	\nonumber D^W_\gamma{F}(w)=\int_{0}^{T}\psi(t,w).g(t)dt.
	\end{equation}
	the variable $F$ is Malliavin differentiable in Wiener space and
	$D^WF=(D^W_t{F})_{0\le{t}\le{T}}:=(\psi(t,w))_{0\le{t}\le{T}}$.
	 We define the set of all $F:\Omega\to\mathbb{R}$ such that $F$ is differentiable by $\mathbb{D}_W^{1,2}$. If we denote by $\mathcal{S}$ the set of all functionals $F=\varphi(\theta_1,\theta_2,...,\theta_n)$ where $\phi$ is a smooth function with bounded derivatives of any order and $\theta_i=\int_{0}^{T}f_i(t)dB_t$ with $f_i\in{L^{2}(\left[0,T\right])}$, Then $F\in\mathbb{D}_W^{1,2}$ and the derivative of $F$ is
	\begin{equation*}
	D^W_t{F}(w)=\sum_{i=1}^{n}\frac{\partial{\varphi}}{\partial x_i}(\theta_1,....,\theta_n)f_i(t).
	\end{equation*}
For every integer $n$ and $p \geq 2$, the space $\mathbb{D}_W^{n,p}$ is the closure of $\mathcal{S}$ with respect to the norm defined by 
\begin{equation*}\label{equ20}
\lVert{F}\rVert_{n,p}^p=\lVert{F}\rVert^p_{L^{p}(\Omega)}+\lVert{(D^W)^n{F}}\rVert^p_{L^{p}(\left[0,T\right]^n\times\Omega)}.
\end{equation*}
The Skorohod operator is the adjoint operator of $D_W$ from $L^{2}([0,T] \times \Omega)$ to $\mathbb{D}_W^{1,2}$. Later, we will use the following  duality relation, which states that for given $F\in\mathbb{D}^{1,2}$ and $u \in Dom(\delta^W)$ 
\begin{equation*}
\mathbb{E}\Big(\left\langle{D^WF,u}\right \rangle_{L^{2}[0,T]}\Big):=\mathbb{E}\Big(\int_{0}^{T}(D^W_tF)u_tdt\Big)=\mathbb{E}\Big(F\delta^W(u)\Big).
\end{equation*}
Also, 
\begin{equation}\label{deltaboundd}
\delta^W (Fu)=F\delta^W(u)-\left\langle{D^WF,u}\right \rangle_{L^{2}[0,T]},  \quad \qquad \mathbb{E}\Big(\delta^W (u)\Big)^2 \le \lVert{F}\rVert^2_{1,2}.
\end{equation}
For every adapted process $u$, $\delta^W(u)$ can be represented by the stochastic integral $\int_0^T u(s)dW_s$.
\subsection{The Malliavin calculus on Poisson space}
There are two different approaches to introduce the Malliavin derivative of Levy processes. One is introduced by the chaos expansion criteria which is not satisfying in the rule chain, and the other is introduced by the closure of the set of Poisson functionals that satisfies the chain rule. We recall some concepts and for more details, we refer to \cite{n3}.\\
\subsubsection{First approach}
Consider a Levy process $N$ with the Levy measure $v$ on a complete separable metric space $(\mathbb{R}_0, \mathcal{B})$. Let $L^{2}([0,T] \times \mathbb{R}_0^{n})$ be the space of symmetric square integrable functions on the $([0,T] \times \mathbb{R}_0^{n}, m \times v \times \cdots \times v)$, where $m$ is an atomless measure on $[0,T]$. Given $h\in L^{2}([0,T] \times \mathbb{R}_0^{n})$ and fixed $z\in \mathbb{R}_0$, we write $h(t,.,z)$ to indicate the function on $\mathbb{R}_0^{n-1}$ given by $(z_1,...,z_{n-1})\to h(t, z_1,...,z_{n-1}, z)$.
Denote the set of random variables $F$ in $L^{2}(\Omega)$ with a chaotic decomposition $F=\sum_{n=0}^{\infty}I_n(h_n)$ by $\mathbb{D}_N^{1,2}$, that $h_n\in L_s^{2}([0,T] \times \mathbb{R}_0^{n})$, satisfying
\begin{equation*}
	\nonumber \sum_{n\ge1}nn!\lVert{h_n}\rVert_{L^{2}([0,T] \times \mathbb{R}_0^{n})}^{2}<\infty.
\end{equation*}
Then, if $F\in \mathbb{D}_N^{1,2}$ we define the Malliavin derivative $D^{N}$ of $F$ as the $L^{2}([0,T] \times \mathbb{R}_0)$-valued random variable given by
\begin{equation*}
	\nonumber D_{t,z}^{N}F=\sum_{n\ge1}nI_{n-1}(h_n(t, .,z)),\;\;z\in \mathbb{R}_0.
\end{equation*}
The operator $D^{N}$ is a closed operator from $\mathbb{D}_N^{1,2}\subset L^{2}(\Omega)$ into $L^{2}(\Omega\times [0,T] \times \mathbb{R}_0)$ and satisfy the following rules.
\begin{lemma}\label{chaind}\cite{n3}
	Let $F,G\in\mathbb{D}_N^{1,2}$ Suppose that $FG\in L^{2}(\Omega)$ and $(F + D^{N}F)(G + D^{N}G)\in L^{2}(\Omega\times [0,T] \times \mathbb{R}_0)$.Then the product FG also belongs to $\mathbb{D}_N^{1,2}$ and
	\begin{equation*}
		\nonumber D_{t,z}^{N}(FG)=FD_{t,z}^{N}G+GD_{t,z}^{N}F+D_{t,z}^{N}FD_{t,z}^{N}G.
	\end{equation*}
\end{lemma}
\begin{proposition}\label{chainf}\cite{n3}
	Let $F$ be a random variable in $\mathbb{D}^{1,2}$ and let $\varphi$ be a real continuous function such that $\varphi(F)$ belongs to $L^{2}(\Omega)$ and $\varphi(F + D^{N}F)$ belongs to $L^{2}(\Omega\times Z)$. Then $\varphi(F)$ belongs to $\mathbb{D}^{1,2}$ and
	\begin{equation}
		\nonumber D_{t,z}^{N}\varphi(F)=\varphi(F+D_{t,z}^{N}F)-\varphi(F).
	\end{equation}
\end{proposition}
It is remarkable that in this approach if $F \in \mathbb{D}^{1,2}$, then the following relation will be also held.
\begin{equation}\label{indicator}
D^N1_{F > k}  = 1_{F+D^N F} -1_{F}, 
\end{equation}
Now, given stochastic process $u$ in $L^{2}(\Omega\times [0,T] \times \mathbb{R}_0)$ admits a unique representation of the following form that for each $(t,z)\in [0,T] \times \mathbb{R}_0$
\begin{equation*}\label{eqa5}
	u(t, z)=\sum_{n\ge0}I_n(h_n(t,.,z)),
\end{equation*}
where the function $h_n \in L^{2}( [0,T] \times\mathbb{R}_0^{n})$. If 
\begin{equation}
	\nonumber \sum_{n\ge0}(n+1)!\lVert{h_n}\rVert_{L(\mathbb{R}_0^{n+1})}^{2}<\infty,
\end{equation}
we say $u$ is in the domain of the divergence operator $\delta^{N}$, denoted by $Dom\delta^{N}$ and 
\begin{equation*}
	\nonumber \delta^{N}(u)=\sum_{n\ge0}I_{n+1}(\tilde{h}_n),
\end{equation*}
where $\tilde{h}_n$ stands for the symmetrization of $h$ as a function in the last $n + 1$ variables. For instance, if $u(z) = h(z)$ is a deterministic function in $L^{2}(\mathbb{R}_0)$ then $\delta(u) = I_1(h)$. If $u(z)=I_1(h(.,z))$, with $h\in L^{2}(\mathbb{R}_0)$, then $\delta(u)=I_2(h)$.\\
The following result characterizes $\delta^{N}$ as the adjoint operator of $D^{N}$.
\begin{proposition}\label{prodelta}\cite{n3}
	If $u\in Dom\delta^{N}$, then $\delta^{N}(u)$ is the unique element of $L^{2}(\Omega)$ such that, for all $F\in\mathbb{D}_N^{1,2}$,
	\begin{equation*}
		\nonumber \mathbb{E}(\left \langle {D^{N}F,u} \right \rangle_{L^{2}([0,T] \times\mathbb{R}_0)})=\mathbb{E}(F\delta^{N}(u)).
	\end{equation*}
	Conversely, if $u$ is a stochastic process in $L^{2}(\Omega\times [0,T] \times \mathbb{R}_0)$ such that, for some $G\in L^{2}(\Omega)$ and for all $F\in\mathbb{D}_N^{1,2}$,
	\begin{equation*}
		\nonumber \mathbb{E}(\left \langle {D^{N}F,u} \right \rangle_{L^{2}([0,T] \times\mathbb{R}_0)})=\mathbb{E}(FG),
	\end{equation*}
	then $u$ belongs to $Dom\delta^N$ and $\delta^{N}(u)=G$.
\end{proposition}
The divergence operator $\delta$ satisfies the following product rule.
\begin{proposition}\label{prochaindelta}\cite{n3}
	Let $F\in\mathbb{D}_N^{1,2}$ and $u\in Dom\delta$ such that the product $uDF$ belongs to $Dom\delta^N$ and the right-hand side of
	\eqref{equ6}
	below belongs to $L^{2}(\Omega)$. Then $Fu\in Dom\delta$ and
	\begin{equation}\label{equ6}
	\delta^N(Fu)=F\delta^N(u)-\left\langle{D^NF,u}\right \rangle_{L^{2}([0,T] \times\mathbb{R}_0)}-\delta^N(uDF).
	\end{equation}
\end{proposition}
\subsubsection{Second approach}
We make use of the notation
\begin{equation*}
	N(h):=\int_{[0,T]}\int_{\mathbb{R}_{0}}h(t,z)N(dt.dz)
\end{equation*}
for every $h\in L^{1}([0,T]\times\mathbb{R}_{0}, m\times v)$. Denote by $C_0^{0,2}([0,T]\times\mathbb{R}_{0})$ the set of continuous functions $h:[0,T]\times\mathbb{R}_{0}\to\mathbb{R}$ that have compact support and are twice differentiable on $\mathbb{R}_{0}$.
We consider the set $\mathcal{S}$ of cylindrical random variables of the form
\begin{equation}\label{equ50}
	F=\varphi(N(h_1),...,N(h_n)),
\end{equation}
where $\varphi\in C_0^{2}(\mathbb{R}^{n})$ and $h_i\in C_0^{0,2}([0,T]\times\mathbb{R}_{0})$ for $1\le i \le n$. It is easy to show that the set $\mathcal{S}$ is dense in $L^{2}(\Omega)$.
The Malliavin derivative of a simple random variable $F$ in $\mathcal{S}$ of the form \eqref{equ50} is defined as the two parameter process
	\begin{equation*}
		D_{t,z}^{N_P}F=\sum_{k=1}^{n}\frac{\partial\varphi}{\partial x_k}(N(h_1),...,N(h_n))\partial_zh_k(t,z),\; \quad (t,z)\in [0,T]\times\mathbb{R}_{0}.
	\end{equation*}
In particular, $D_{t,z}^{N_P}(N(h)) = \partial_zh$. 
Define the scalar product $<.,.>$ for every $u,\tilde{u} \in L^2(\Omega)$ as 
\begin{equation*}
	<u,\tilde{u}>_N:= \int_0^T \int_{\mathbb{R}_0} u(s,z)\tilde{u}(s,z) N(ds,dz),
\end{equation*}
and denote $\Vert . \Vert_N$ as its associated norm. Also, let $\mathbb{D}_{N_p}^{1,p}$, for every $p \geq 1$, the closure of $\mathcal{S}$, as the domain of the operator $D_{t,z}^{N_P}$, with respect to the seminorm
\begin{equation*}
	\Vert F \Vert_{1,N}^p := \mathbb{E}(|F|^p)+\mathbb{E}(\Vert D^{N_P}F \Vert^p_N).
\end{equation*} 
The next result is the chain rule for the Malliavin derivative in the Poisson framework.
\begin{proposition}\cite{n3}\label{prod}
	Let $\varphi$ be a function in $C^{1}(\mathbb{R})$ with bounded derivative, and let $F$ be a random variable in $\mathbb{D}_{N_p}^{1,2}$. Then, $\varphi(F)$ belongs to $\mathbb{D}_{N_p}^{1,2}$ and
	\begin{equation*}
		D_{t,z}^{N_P}(\varphi(F))=\varphi^{\prime}(F)D_{t,z}^{N_P}(F).
	\end{equation*}
\end{proposition}
The authors in \cite{song} have stated a powerful tool called integration by parts formula for this type of derivative in the following form in some Sobolev spaces we recall here. For every $p \geq 1$, denote by $\mathbb{L}_p$ the set of all predictable processes $\psi$ on $[0,T]\times \mathbb{R}_0$ with finite norm
\begin{equation*}
\Vert \psi\Vert_{\mathbb{L}_p}=\Big[ \mathbb{E}\Big(\int_{\mathbb{R}_{0}}\int_0^T \psi(s,z) m(ds)\nu (dz)\Big)^p\Big]^{\frac{1}{p}}+\Big[ \mathbb{E}\Big(\int_{\mathbb{R}_{0}}\int_0^T \psi^p(s,z) m(ds)\nu(dz)\Big)\Big]^{\frac{1}{p}},
\end{equation*}
and denote by $\mathbb{V}_p$ the set of all predictable processes $\psi$ on $[0,T]\times \mathbb{R}_0$ with finite norm
\begin{equation*}
\Vert \psi\Vert_{\mathbb{V}_p}=\Vert \frac{ \partial\psi}{\partial z}\Vert_{\mathbb{L}_p}+\Vert \rho\psi\Vert_{\mathbb{L}_p},
\end{equation*} 
where $\rho(z)=\vert z \vert^{-1}$.  We shall write $\mathbb{V}_{\infty}:= \bigcap_{p \geq 1}\mathbb{V}_p$.
\begin{proposition}\label{bypart}
Given $F \in \mathbb{D}_{N_p}^{1,p}$, for $p\geq 2$, and $w_0 \in \mathbb{V}_{\infty}$  we have 
\begin{equation*}
 \mathbb{E}\Big(<D^{N_P}(F), w_0>_N\Big)= \mathbb{E}\Big(F\int_{\mathbb{R}_{0}}\int_0^T  \frac{1}{\theta} \frac{ \partial (\theta w_0)}{\partial z}(t,z) \tilde{N}(dt,dz)\Big),
\end{equation*}
where $\nu(dz)=\theta(z) dz$.
\end{proposition}

\section{Stochastic jump processes with stochastic intensity}\label{sec4}
In this section, we recall the concept of stochastic intensity as desired by B{\'e}rmaud in Chapter 5 of \cite{bremand} and introduce the model, state the assumptions, and present the key lemmas required for the main results. \\
Let $(\Omega,\mathcal{F},\mathit{P})$ be a Wiener-Poisson space with a risk neutral probability $\mathit{P}$.
Assume that $N_t$ is a Poisson process and $\mathcal{F}_t^N$ is an $\sigma$-field generated by $N$ with the density of jumps sizes $C_z$, as $z \in \mathbb{R}_0$ and stochastic intensity process $\lambda$. For given  $\sigma$-field $\mathcal{F}_t$, the process $\lambda_t$ is an $\mathcal{F}_t$-intensity of $N_t$  if  for every $s,t \in [0,T]$
\begin{equation*} 
	\mathbb{E}\Big(\int_{\mathbb{R}_0}\int_t^s  N(du,dz)\vert \mathcal{F}_t\Big)=\mathbb{E}\Big(\int_{\mathbb{R}_0}\int_t^s C_z\lambda_u dudz\vert \mathcal{F}_t\Big),
\end{equation*}
and so that  $\tilde{N}(t,z)=N(t.z)-\int_{\mathbb{R}_0}\int_0^tC_z\lambda_sdsdz$  is an $\mathcal{F}_t$-martingale. Also, obviously, for every $0 \leq t,s \leq T$ and for every $\mathcal{F}_t$-predictable function $k$  
\begin{equation*} 
	\mathbb{E}\Big(\int_{\mathbb{R}_0}\int_t^s k(u,z) N(du,dz)\vert  \mathcal{F}_t \Big)=\mathbb{E}\Big(\int_{\mathbb{R}_0}\int_t^s k(u,z)C_z\lambda_u dudz\vert  \mathcal{F}_t \Big). 
\end{equation*}
We refer the reader to Chapter 5 of \cite{bremand} for more details. It is worth mention that one can easily show  \cite{bremand} that if $\lambda$ is $\mathcal{G}$-measurable, for every measurable function $k$ such that $\mathbb{E}\Big(\int_{\mathbb{R}_0}\int_0^t (k(s,z))^2 \lambda_s dsC_zdz\Big)< \infty$, 
\begin{equation}\label{expglambda}
	\mathbb{E}\Big(exp\{iu\int_{\mathbb{R}_0}\int_0^t k(s,z) N(ds,dz) \}\Big \vert \mathcal{G}\Big)=exp\Big\{\int_{\mathbb{R}_0}\int_0^t (e^{iuk(t,z)}-1) \lambda_s C_zdsdz\Big\}.
\end{equation}
In this manuscript, we assume that the underlying asset price $S=(S_{t})_{t\in[0,T]}$ with the jump stochastic intensity process $\lambda=(\lambda_{t})_{t\in[0,T]}$ of Poisson process $N_{t}$ can be governed by the following system of SDEs:
\begin{equation}\label{equ1}
\begin{cases} 
dS_{t}&=\mu S_{t}dt+\sigma_{1}S_{t}dW_{t}^S+\int_{\mathbb{R}_{0}}(e^{J_{t,z}}-1)S_{t}\tilde{N}(dt,dz), \\
d\lambda_{t}&=\kappa(\Theta-\lambda_{t})dt+\sigma_{2}\sqrt{\lambda_{t}}dW_{t},
\end{cases}
\end{equation}
where $(W_{t})_{t\in[0,T]}$ and  $(W_{t}^S)_{t\in[0,T]}$ are independent Brownian motions, $N_{t}$ is independent of $W_t^S$, $\mu$ denotes the riskless interest rate, $J$ is a cadlag function, the mean-reverting speed parameter $\kappa, \sigma_{2}$ and $\sigma_{1}$ are positive constants and the long term mean $\Theta$ satisfying $2\kappa\Theta\; > \sigma_{2}^2$. \\
For $\mathcal{F}_t^\lambda$, $\sigma$-field generated by $\lambda$, let $\mathcal{F}_t= \mathcal{F}_t^N \vee \mathcal{F}_t^\lambda$ and $\mathcal{G}= \mathcal{F}_t^\lambda$. In this case, obviously,  $\mathcal{F}_t= \mathcal{F}_t^N$ and for every $0 \leq t,s \leq T$ and for every $\mathcal{F}_t$-predictable function $k$  
\begin{equation*} 
	\mathbb{E}\Big(\int_{\mathbb{R}_0}\int_t^s k(u,z) N(du,dz)\vert  \mathcal{F}_t^\lambda \Big)=\int_{\mathbb{R}_0}\int_t^s k(u,z)C_z\lambda_u dudz. 
\end{equation*} \\
We also assume the following conditions throughout the paper.\\
{\bf Condition  H1:} \begin{itemize}
\item For every $p \geq 1$ and for almost everywhere $0\leq t \leq T$ 
\begin{equation}\label{assum1}
   	  \int_{\mathbb{R}_{0}}e^{pJ_{t,z}}C_zdz  = u_{p} < \infty,       \qquad   |v_t|:=\Big|\int_{\mathbb{R}_{0}}(e^{J_{t,z}}-1) C_zdz \Big|\geq \epsilon_0 > 0.  
\end{equation}
\item For $p_0 := max\{p \geq 2 ; ~~ u_{p} \sigma_2^2 \leq 2k \}$, we assume that $p_0 \geq 32$.
\end{itemize}
\begin{remark}
The second part of condition {\bf H1} will guarantee the boundedness of $p$-moment of the solution $S_t$, for every $2 \le p \leq p_0$. Here, we note that it is not limitation at all.
For instance, consider the probability density function (PDF) of the Double Exponential Distribution 
\begin{equation*}
	f_J(z) = 
	\begin{cases} 
		\frac{1}{p_{\text{u,z}}} \cdot \eta_2 e^{\eta_2 z}, & \text{if } z < 0, \\
		\frac{1}{1 - p_{\text{u,z}}} \cdot \eta_1 e^{-\eta_1 z}, & \text{if } z \geq 0.
	\end{cases}
\end{equation*}
 The parameter $\eta_1$ is the scale parameter for the positive exponential part, $\eta_2=5$ is the scale parameter for the negative exponential part and $p_{\text{u,z}} \in [0, 1]$. Let $J_{t,z}=-\vert z \vert$, for every $z \in \mathbb{R}_0$. Then for every $p \geq 2$
\begin{equation*}
	u_p=\int_{0}^{\infty}\frac{1}{1-p_{u,z}}e^{-pz}\eta_1e^{-z\eta_1}dz+\int_{-\infty}^{0}\frac{1}{p_{u,z}}e^{pz}\eta_2e^{z\eta_2}dz=\frac{1}{1-p_{u,z}}(\frac{\eta_1}{\eta_1+p}+\frac{1}{p_{u,z}}\frac{\eta_2}{\eta_2+p})=1.
\end{equation*}
and $p_0=\infty$ when $\sigma_{2}^{2}\leq2\kappa$.\\
In addition, for the jumps of the Gaussian distribution with the mean 0 and the variance $\sigma$, utilizing $J_{t,z}= {-z}1_{z \ge 0}+z1_{z<0}$ we deuce 
\begin{equation*}
	u_p=\int_{0}^{\infty}\frac{1}{\sqrt{2\pi }\sigma}e^{-pz}e^{-\frac{z^2}{2\sigma^2}}dz+\int_{-\infty}^{0}\frac{1}{\sqrt{2\pi }\sigma}e^{pz}e^{-\frac{z^2}{2}}dz= 2 \int_{0}^{\infty}\frac{1}{\sqrt{2\pi }\sigma}e^{-p{z}}e^{-\frac{z^2}{2\sigma^2}}dz\leq 1.
\end{equation*}
Therefore, $p_0=\infty$ when $\sigma_{2}^{2}\leq2\kappa$.\\
\end{remark}
We know that the solution to the stochastic differential \eqref{equ1} is as follows, see \cite{n22}.
\begin{align*}
\nonumber S_{t}&=S_{0}\exp\Big\{(\mu-\frac{\sigma_{1}^{2}}{2})t+\sigma_{1}W_{t}^S+\int_{0}^{t}\int_{\mathbb{R}_{0}}(J_{s,z}-e^{J_{s,z}}+1)C_z\lambda_sdzds\\
&+\int_{\mathbb{R}_{0}}\int_{0}^{t}J_{s,z}\tilde{N}(ds,dz)\Big\}\\
& =:S_0\exp\Big\{X_t\Big\}=: S_0Y_t \exp\Big\{(\mu-\frac{\sigma_{1}^{2}}{2})t+\sigma_{1}W_{t}^S\Big\}, 
\end{align*}
where $Y_t$ satisfying 
\begin{equation}\label{y}
dY_t=Y_t(e^{J_{t,z}}-1) \tilde{N}(dt,dz), \qquad Y_0=1.
\end{equation} 
To approximate the $p$-moments of $S_t$ we need the following lemmas. It is noteworthy that in \cite{altmayer}, the authors demonstrated that, using the Ito formula, for every $s \leq t$
\begin{equation}\label{suplambda}
\mathbb{E}(\sup_{0 \leq t \leq T} \lambda_t ^p) <\infty  \quad  \forall p \geq 1,    \quad  and   \quad  \sup_{0 \leq t \leq T}\mathbb{E}(\lambda_t ^{-p}) <\infty, \quad \forall p \geq 1 ~ s.t.~ 2\kappa\theta > p\sigma_2^2,
\end{equation}
\begin{lemma}\label{gamm}
For every function $\gamma_t$ such that $\frac{\partial }{\partial t}\gamma_t -\kappa \gamma_t +\frac12 \sigma_2^2 \gamma_t^2 \le 0$, we  have
$$\mathbb{E}(  e^{\int_0^t \gamma_t \lambda_s ds}) \leq  \int_0^t e^{\gamma_0\lambda_{0}}e^{\kappa\Theta\int_0^u \gamma_s ds} du. $$ 
\end{lemma}
\begin{proof}
We use Ito formula to have
\begin{align}
\nonumber
de^{\gamma_t\lambda_{t}}&=(\frac{\partial }{\partial t}\gamma_t) \lambda_t e^{\gamma_t\lambda_{t}} dt+\gamma_t\kappa (\Theta-\lambda_{t})e^{\gamma_t\lambda_{t}}+\sigma_{2}\gamma_t e^{\gamma_t\lambda_{t}}\sqrt{\lambda_{t}}dW_t+\frac{\sigma_{2}^{2}}{2}\gamma_t^{2}e^{\gamma_t \lambda_{t}}\lambda_{t}dt\\
\nonumber
&=\kappa\Theta\gamma_t e^{\gamma_t \lambda_{t}}dt+(\frac{\partial }{\partial t}\gamma_t+\frac{\gamma_t^{2}\sigma_{2}^{2}}{2}-\kappa\gamma_t)\lambda_{t}e^{\gamma_t\lambda_{t}}dt+\sigma_{2}\gamma_t e^{\gamma_t \lambda_{t}}\sqrt{\lambda_{t}}dW_t.
\end{align}
Taking the expectation on both sides, our assumption, and applying Gronwall inequality deduce 
\begin{equation*}
\mathbb{E}(e^{\gamma_t \lambda_{t}})\leq e^{\gamma_0\lambda_{0}}e^{\kappa\Theta\int_0^t \gamma_s ds},
\end{equation*}
and then Yensen inequality completes the result. 
\end{proof}
\begin{corollary}
For every $2 \le p \le \frac12 p_0$, the following inequality satisfies.
 $$ \sup_{0 \leq t \le T} \mathbb{E}(  e^{u_p \int_0^t \lambda_s ds}) < \infty. $$ 
 \end{corollary}
 Consequently, we conclude the moments of $Y_t$ and $S_t$ as follows. 
\begin{lemma}\label{lemy}
The solution $Y_t$ has uniformly bounded $p$-moments for every $2 \le p \leq \frac12p_0$ such that $p\sigma_2^2 < \kappa$. In addition, 
if $(16u_1 \vee 1)\sigma_2^2 < 2\kappa$, the inverse of $Y_t$ has uniformly bounded fourth moments. 
\end{lemma}
\begin{proof}
According to the definition of $Y_t$, Cauchy-Schwartz inequality, and \eqref{expglambda}
\begin{align*}
\mathbb{E}(Y_t^{p})&=\mathbb{E}^\frac12\Big(e^{2p\int_{\mathbb{R}_0}\int_{0}^{t} J_{s,z}N(ds,dz)}\Big)\mathbb{E}^\frac12\Big(e^{2p\int_{0}^{t}\int_{\mathbb{R}_0}(1-e^{J_{s,z}})C_z\lambda_sdzds}\Big)\leq\mathbb{E}^{\frac{1}{2}}\Big(e^{\int_{0}^{t}\int_{\mathbb{R}_{0}}(e^{2pJ_{s,z}}-1)C_z\lambda_sdzds}\Big)\times\mathbb{E}^{\frac{1}{2}}\Big(e^{2p\int_{0}^{t}\lambda_sds}\Big)  < \infty, 
\end{align*}		
where we used Lemma \ref{gamm} in the last inequality by using $\gamma_t = u_{2p}$ and $\gamma_t =2p$.\\
Similarly, due to the fact there exists some positive constant $c_0 \leq 1$ that $1- e^{-x} \geq -c_0x$, for every $x \in \mathbb{R}$ and $\vert x \vert \leq 1$, we deduce 
\begin{align*}
\mathbb{E}(Y_t^{-8})&=\mathbb{E}^\frac12\Big(e^{-16\int_{\mathbb{R}_0}\int_{0}^{t} J_{s,z}N(ds,dz)}\Big)\mathbb{E}^\frac12\Big(e^{-16\int_{0}^{t}\int_{\mathbb{R}_0}(1-e^{J_{s,z}})C_z\lambda_sdzds}\Big)\\
&\leq\mathbb{E}^{\frac{1}{2}}\Big(e^{\int_{0}^{t}\int_{\mathbb{R}_{0}}(e^{-16J_{s,z}}-1)C_z\lambda_sdzds}\Big)\times\mathbb{E}^{\frac{1}{2}}\Big(e^{16u_1\int_{0}^{t}\lambda_sds}\Big) \\
& \leq \mathbb{E}^{\frac{1}{2}}\Big(e^{\int_{0}^{t}\int_{| 8J_{s,z} | \le 1}(-16c_0J_{s,z})C_z\lambda_sdzds}\Big)\times\mathbb{E}^{\frac{1}{2}}\Big(e^{16u_1\int_{0}^{t}\lambda_sds}\Big)  \\
&\leq \mathbb{E}^{\frac{1}{2}}\Big(e^{c_0\int_{0}^{t}\int_{\vert 16J_{s,z} \vert \leq 1}c_z\lambda_sdzds}\Big)\times\mathbb{E}^{\frac{1}{2}}\Big(e^{16u_1\int_{0}^{t}\lambda_sds}\Big)  < \infty.
\end{align*}		
\end{proof}
As a consequence, we easily see that for $0\le p \le \frac14p_0$ such that $(16u_1 \vee 1)\sigma_2^2 < \kappa$,
\begin{align}
\mathbb{E}\Big(\sup_{0 \leq t \leq T}\vert S_t\vert^{-4}\Big) \leq S_0^{-p}\mathbb{E}\Big(\sup_{0 \leq t \leq T}\vert Y_t\vert^{-8}\Big)\mathbb{E}\Big(\sup_{0 \leq t \leq T}e^{-8(\mu-\frac{\sigma_{1}^{2}}{2})t+\sigma_{1}W_{t}^S)}\Big) < \infty. \label{sinverse}
\end{align}

\begin{lemma}\label{pp0}
	The solution $S_t$ of  \eqref{equ1} is unique and uniformly is in $\bigcap_{2 \leq p \leq \frac14 p_0, 2p\sigma_2^2 < \kappa}L^p(\Omega)$, i.e., for every $2 \leq p \leq \frac14 p_0$ such that  $2p\sigma_2^2 < \kappa$ we have
\begin{equation}
\nonumber\mathbb{E}\Big(\sup_{0 \leq t \leq T}\vert S_t\vert^{p}\Big) < \infty.
\end{equation}
\end{lemma}
\begin{proof}
We know that for any $p\ge2$, 
$$\mathbb{E}\Big(\sup_{t\in\left[0,T\right]} \exp\Big\{p(\mu-\frac{\sigma_{1}^{2}}{2})t+p\sigma_{1}W_{t}^S\Big\}) < \infty.$$
So, it is sufficient to show that equation \eqref{y} has a unique solution. To do this, with the same proof of Lemma 2.3. in \cite{regularitysong} and Section 5.1.1 of \cite{n21}, we derive that for any $p\ge2$ and every step time $h$, there exists a constant $C^0_p>0$ such that:
\begin{align}
\mathbb{E}\biggl(\sup_{s\in\left[t,t+h\right]}\bigg|\int_{\mathbb{R}_{0}}\int_{t}^{s}&(e^{J_{u,z}}-1)Y_u \tilde{N}(du,dz)\bigg|^{p}\vert\mathcal{F}_{t}^\lambda\biggr)\nonumber\\
 &\le{C^0_p}\mathbb{E}\biggl(\bigg[\int_{t}^{t+h}\int_{\mathbb{R}_{0}}Y_u^{2}(e^{J_{u,z}}-1)^2 C_z\lambda_u dzdu\bigg]^{\frac{p}{2}}\vert\mathcal{F}_{t}^\lambda\biggr)\nonumber\\
&+{C^0_p}\mathbb{E}\biggl(\bigg[\int_{t}^{t+h}\int_{\mathbb{R}_{0}}Y_u^{p}(e^{J_{u,z}}-1)^p C_z\lambda_u dzdu\vert\mathcal{F}_{t}^\lambda\biggr).\label{supsup}
\end{align}
Define the new probability measure $p_1(A)= \frac{\int_A\lambda_s ds}{ \int_t^{t+h} \lambda_s ds }$, for every $A \subset [t,t+h]$ as $1_A$ is the indicator function, and applying Yensen inequality to result 
 \begin{align*}
\mathbb{E}\biggl(\sup_{s\in\left[t,t+h\right]}Y_s ^p\vert\mathcal{F}_{t}^\lambda\biggl)
 &\le\mathbb{E}\big(Y_t^p\vert\mathcal{F}_{t}^\lambda\big) \nonumber\\
&+{C^0_p}\mathbb{E}\biggl(\big(\int_t^{t+h} \lambda_s ds\big)^{\frac{p}{2}-1}\int_{t}^{t+h}\int_{\mathbb{R}_{0}}Y_u^{p}(e^{J_{u,z}}-1)^p C_z\lambda_u dzdu\vert\mathcal{F}_{t}^\lambda\biggr)\nonumber \\
&+{C^0_p}\mathbb{E}\biggl(\int_{t}^{t+h}\int_{\mathbb{R}_{0}}Y_u^{p}(e^{J_{u,z}}-1)^p C_z\lambda_u dzdu\vert\mathcal{F}_{t}^\lambda\biggr)\nonumber\\
 &\le \mathbb{E}\big(Y_t^p\vert\mathcal{F}_{t}^\lambda\big)+2^p{C_p}{u_p}\sup_{t \leq s \leq t+h}\mathbb{E}^{\frac12}\biggl( Y_s^{2p}\vert\mathcal{F}_{t}^\lambda\biggr)\mathbb{E}^{\frac12}\big({ \int_t^{t+h} \lambda_s ds}\vert\mathcal{F}_{t}^\lambda\big)^{p}\nonumber\\
&+2^p{C^0_p}u_p\sup_{t \leq s \leq t+h}\mathbb{E}^{\frac12}\biggl( Y_s^{2p}\vert\mathcal{F}_{t}^\lambda\biggr)\mathbb{E}^{\frac12}\big({ \int_t^{t+h} \lambda_s ds}\vert\mathcal{F}_{t}^\lambda\big)^2. 
\end{align*}
Finally,  Lemma \ref{lemy} and \eqref{suplambda} complete the proof. 
\end{proof}
In the last part of this section, we note that getting the partial derivatives of $S_t$ with respect to $S_0$ shows that the stochastic flow of $S_t$ exists and it is  
\begin{align}
	\frac{\partial{S_t}}{\partial{S_0}}=\frac{S_t}{S_0}=\exp\{X_t\}=Y_t \exp\{(\mu-\frac{\sigma_{1}^{2}}{2})t+\sigma_{1}W_{t}\}. \label{flow1}
\end{align}
Therefore, this flow is in $L^{p}$-space for every $2 \leq p \leq \frac14 p_0$ and $2p \sigma_2^2 < \kappa$.  
\subsection{Malliavin derivative of the solution on Wiener space}
In this section, we obtain the Malliavin derivative of the solution $S_t$ and we will also consider some Skorokhod integrable directions in which the inverse of directional derivatives are in $L^p(\Omega)$,  for all $2 \leq p \leq p_0$.\\
Due to the representation of the solution with respect to $X_t$, we need to find its derivative. Gaussian Malliavin derivative of $X_t$ comes as follows:
\begin{equation}\label{equ10}
D_{u}^{W}X_t= -\int_{u}^{t}D_u^{W}\lambda_s\int_{\mathbb{R}_{0}}(e^{J_{s,z}}-1)C_zdzds=:-\int_{u}^{t}v_u D_u^W\lambda_{s}ds,
\end{equation}
In \cite{altmayer}, the authors have shown that using the Ito formula and taking the Malliavin derivative with respect to the Brownian motion,  for every $s \leq t$
\begin{equation}\label{14}
D_s^{W}\lambda_t=\sigma_{2}\sqrt{\lambda_t}1_{0\leq s\leq t}\exp\Big\{-{\int_{s}^{t}(\frac{\kappa}{2}+\frac{C_{\sigma}}{\lambda_r})dr}\Big\},
\end{equation}
where $C_{\sigma}=\frac{\kappa\Theta}{2}-\frac{\sigma_{2}^{2}}{8}$ is a positive number. Here, we represent that the inverse of directed Malliavin derivative of $X_t$ in some directions, which are also in the domain of the Skorokhod operator, can belong to all $L^{p}$ spaces for any $p \geq 2$.
\begin{theorem}\label{th1}
When $2\kappa\theta > 3\sigma_2^2$, there exists a direction $h(.) \in dom(\delta^W)$, defined as the following 
\begin{equation*}
 h(u)=\frac{1}{v_u}(\frac{\kappa}{2}+\frac{C_{\sigma}}{\lambda_u}), \qquad 0\leq u \leq T,
\end{equation*}
such that 
$\mathcal{B}_T=\left\langle {D_.^{W}X_T,h(.)}\right\rangle$  is almost surely  invertible and 
\begin{equation*}
\Big(\left\langle D_.^WX_T,h(.)\right \rangle\Big)^{-1}\in\bigcap_{2 \leq p}L^p.
\end{equation*}
\end{theorem}
\begin{proof}
From \eqref{equ10}, Fubini theorem and the expression \eqref{14}, we derive 
\begin{align}
 \mathcal{B}_T &=
-\int_{0}^{T}\left\langle v_.{D_.^{W}\lambda_s,h(.)}\right\rangle ds\nonumber\\
&=-\int_0^T\Big[\sigma_{2}\sqrt{\lambda_{s}}exp\{-\int_{0}^{s}(\frac{\kappa}{2}+\frac{C_{\sigma}}{\lambda_r})dr\}\int_{0}^{s}(\frac{\kappa}{2}+\frac{C_{\sigma}}{\lambda_u})  exp\{\int_{0}^{u}(\frac{\kappa}{2}+\frac{C_{\sigma}}{\lambda_r})dr\}du \Big]ds \nonumber\\
&= - \int_0^T\Big[\sigma_{2}\sqrt{\lambda_{s}}(1-exp\{-\int_{0}^{s}(\frac{\kappa}{2}+\frac{C_{\sigma}}{\lambda_r})dr\})\Big]ds. \label{equ17}
\end{align}
Applying the Gamma function results
\begin{align*}
 \mathbb{E}(\frac{1}{\vert \mathcal{B}_T\vert^{p}})=&\frac{1}{\Gamma(p)}\mathbb{E}\big(\int_{0}^{\infty}z^{p-1}e^{-z\int_{0}^{T}\Big[\sigma_{2}\sqrt{\lambda_{s}}(1-exp\{-\int_{0}^{s}(\frac{\kappa}{2}+\frac{C_{\sigma}}{\lambda_r})dr\})\Big]ds}dz\big)\\
&\leq \frac{1}{\Gamma(p)}\mathbb{E}\big(\int_{0}^{\infty}z^{p-1}e^{-z(1-e^{-\frac{T\kappa}{2}})\int_{0}^{T}\sigma_{2}\sqrt{\lambda_{s}}ds}dz\big)\\
&=\frac{1}{\sigma_{2}^p(1-e^{-\frac{T\kappa}{2}})^{p}}\mathbb{E}\big(\frac{1}{(\int_{0}^{T}\sqrt{\lambda_{s}}ds)^p}\big) <\infty,
\end{align*}
thanks  to Lemma 5.2. of \cite{altmayer} in the last inequality. So, we conclude that for every $p \geq 2$, $\mathcal{B}_T^{-1} \in L^P(\Omega)$. 
In the sequel, we will show $h(.) \in Dom(\delta^W)$. According to Proposition 1.3.1 in \cite{nualart}, it is sufficient to show that $h(.) \in \mathbb{D}_W^{1,2}$. To do that, 
from these facts that for every $x,y \in \mathbb{R}$, $(x+y)^2\leq 2x^2+2y^2$ we result
\begin{align}\label{equ23}
\nonumber\mathbb{E}(\int_{0}^{T}&h^{2}(t)dt )+\mathbb{E}(\int_{0}^{T}(D^Wh)^{2}(t)dt )\\
&\le \kappa^2 \int_0^T \frac{1}{2v_u^2} du +2 C_\sigma^2 \int_0^T  \frac{1}{v_u^2 } \mathbb{E}(\frac{1}{\lambda_u^2}) du +\sigma_2^2 C_\sigma^2 \int_0^T  \frac{1}{v_u^2 } \mathbb{E}(\frac{1}{\lambda_u^3}) du  < \infty,
\end{align}
 where we used form (\ref{suplambda}) in the last inequality. 
\end{proof}
\section{Pricing and Delta calculation}\label{sec5}
In this section, we discuss the pricing of the payoff function by weighted Malliavin described in the previous section. We also present an explicit formula to calculate the delta Greek. To do this, we state a representation of the delta as a combination of the Wiener-Malliavin weight and the Poisson-Malliavin weight. We assume the following conditions on the payoff functions.\\
{\bf Condition H2:} The payoff function $\mathit{f}:\mathbb{R^{+}}\to\mathbb{R^{+}}$ is a measurable function with at most polynomial growth $\frac{p_0}{32}$,
$$\vert f(x) \vert \leq c_f (1+\vert x \vert^p) \quad x\in \mathbb{R}^+, \quad p \leq p_0,$$
and locally Riemann integrable, possibly, having discontinuities of the
first kind.\\
Let us introduce the following notations presented in 
\cite{n17}: for every $x\geq 0$,
\begin{equation*}
 F(x)=\int_{0}^{x}\mathit{f}(z)dz, \qquad 	\mathit{g(y)}=\mathit{f(e^{y})},  \qquad  G(y)=\int_{0}^{y}\mathit{g}(z)dz.
\end{equation*}
In this notation, we have  $\mathbb{E}\mathit{f}(S_T)=\mathbb{E}\mathit{g}(X_T)$ and 
		\begin{equation*}
		G(x)=\frac{F(e^{x})}{e^{x}}+\int_{0}^{x}\frac{F(e^{y})}{e^{y}}dy-F(1).
		\end{equation*}	
\begin{theorem}\label{thm1}
	Under condition {\bf{H2}}, the price of a simple derivative can be represented as
	\begin{equation}\label{equ12}
	\mathbb{E}\Big(\mathit{f}(S_T)\Big)=\mathbb{E}\Big(\frac{F(S_T)}{S_T}(1+{Z_T})\Big)=\mathbb{E}\Big(G(X_T)Z_T\Big),
	\end{equation}
where $Z_T=\delta^W(\frac{h(.)}{\mathcal{B}_T})$.
\begin{proof}
Suppose that the function $\mathcal{K}$ is a locally Lipschitz function with $\mathcal{K}^{\prime}(x)=k(x)$ almost everywhere with respect to the Lebesgue measure. Assume additionally that $k$ is of exponential growth and  $\mathcal{K}(X_T)\in \mathbb{D}_W^{1,2}$.
Namely, the Skorokhod integral is the adjoint operator to the Malliavin derivative, therefore
		\begin{align}\label{equ36}
		\nonumber \mathbb{E}\Big(k(X_T)\Big)&=\mathbb{E}\Big(\int_{0}^{T}k(X_T)\frac{D_u^{W}X_T h(u)}{\left\langle{D_.^{W}X_T,h(.)}\right\rangle} du\Big)\\
		\nonumber&=\mathbb{E}\Big(\int_{0}^{T}\frac{D_u^{W}\mathcal{K}(X_T) h(u)}{\left\langle{D_.^{W}X_T,h(.)}\right\rangle}du\Big)=\frac{1}{T}\mathbb{E}\Big(\int_{0}^{T}\mathcal{K}(X_T)\frac{h(.)}{\left\langle{D_.^{W}X_T,h(.)}\right\rangle}dW_u\Big)\\
	 &=\mathbb{E}\Big(\mathcal{K}(X_T)\int_{0}^{T}\frac{h(u)}{\left\langle{D_.^{W}X_T,h(.)}\right\rangle}dW_u\Big)=\mathbb{E}\Big(\mathcal{K}(X_T)Z_T\Big).
		\end{align}
		In particular, for the function $G$ which is a locally Lipschitz function and $g$ is of exponential growth,
				we rewrite \eqref{equ36} for $k=g$ as follows:
		\begin{align}
		\nonumber \mathbb{E}(f(S_T))&=\mathbb{E}(g(X_T))=\mathbb{E}(G(X_T)Z_T)\\
		\nonumber &=\mathbb{E}\Big(\Big(\frac{F(S_T)}{S_T}+\int_{0}^{X_T}\frac{F(e^{y})}{e^{y}}dy-F(1)\Big)Z_T\Big)\\
		\nonumber &=\mathbb{E}\Big(\frac{F(S_T)}{S_T}Z_T\Big)+\mathbb{E}\Big(Z_T\int_{0}^{X_T}\frac{F(e^{y})}{e^{y}}dy\Big)-\mathbb{E}\Big(F(1)Z_T\Big)\\
		\nonumber &=\mathbb{E}\Big(\frac{F(S_T)}{S_T}Z_T\Big)+\mathbb{E}\Big(\int_{0}^{X_T}\frac{F(e^{y})}{e^{y}}dyZ_T\Big),
		\end{align}
		Applying equation
		\eqref{equ36}
		 to $k(x)=\frac{F(e^{x})}{e^{x}}$, we get that
		\begin{equation}
		\nonumber \mathbb{E}\Big(\int_{0}^{X_T}\frac{F(e^{y})}{e^{y}}dyZ_T\Big)=\mathbb{E}\Big(\frac{F(e^{X_T})}{e^{X_T}}\Big)=\mathbb{E}\Big(\frac{F(S_T)}{S_T}\Big).
		\end{equation}
		Hence
		\begin{equation}
		\nonumber\mathbb{E}\Big(f(S_T)\Big)=\mathbb{E}\Big(\frac{F(S_T)}{S_T}Z_T\Big)+\mathbb{E}\Big(\frac{F(S_T)}{S_T}\Big)=\mathbb{E}\Big(\frac{F(S_T)}{S_T}(1+{Z_T})\Big).
		\end{equation}
	\end{proof}
\end{theorem}
\subsubsection{Delta with Wiener-Malliavin weight}
Now, we are ready to present an explicit formula to calculate the Delta Greek. To do this, we state a representation of the delta as a combination of the Wiener-Malliavin weight and the Poisson-Malliavin weight.  
 \begin{theorem}\label{le2}
	Under condition {\bf{H2}}, the delta display with respect to the Wiener -Malliavin weight as 
	\begin{align}
	\nonumber \Delta^W=&\frac{\partial}{\partial{s}}\mathbb{E}\Big(f(S_T)\Big)=\mathbb{E}\Big(f(S_T)\frac{Z_T}{S_0}\Big).
\end{align}
\begin{proof}
	From the fact that $\frac{\partial{Z_T}}{\partial{S_0}}=0$, we derive 
	\begin{align}
		\nonumber \Delta^W&=\frac{\partial}{\partial{S_0}}\mathbb{E}\Big(f(S_T)\Big)=\frac{\partial}{\partial{S_0}}\mathbb{E}\Big(\frac{F(S_T)}{S_T}(1+{Z_T})\Big)\\
		\nonumber &=\mathbb{E}\Big(\frac{\frac{\partial{F(S_T)}}{\partial{S_0}}S_T-F(S_T)\frac{\partial{S_T}}{\partial{S_0}}}{S_T^{2}}(1+{Z_T})+\frac{F(S_T)}{S_T}\frac{\partial{Z_T}}{\partial{S_0}}\Big)\\
		\nonumber &=\mathbb{E}\Big((\frac{F^{\prime}(S_T)}{S_T}-\frac{F(S_T)}{{S_T}^{2}})\frac{\partial{S_T}}{\partial{S_0}}(1+{Z_T})\Big)\\
		\nonumber&=\mathbb{E}\Big(\frac{f(S_T)}{S_T}\frac{{S_T}}{{S_0}}(1+{Z_T})\Big)-\mathbb{E}\Big(\frac{F(S_T)}{{S_T}^{2}}\frac{{S_T}}{{S_0}}(1+{Z_T})\Big)\\
		\nonumber&=\mathbb{E}\Big(\frac{f(S_T)}{S_0}(1+{Z_T})\Big)-\mathbb{E}\Big(\frac{F(S_T)}{S_T}\frac{1}{{S_0}}(1+{Z_T})\Big)\\
		\nonumber&=\mathbb{E}\Big(\frac{f(S_T)}{S_0}{Z_T}\Big),			
\end{align}
where we used Theorem \ref{thm1} in the last equality. 
\end{proof}
\end{theorem}
\subsubsection{Delta with Poisson-Malliavin weight}
We will use the literature in \cite{n18} and calculate the delta with a Malliavin weight regarding the Poisson random measure in two approaches.\\
{\bf In the first approach.}\\
 Due to Proposition \ref{chainf}, we know that $D_{u,z}^{N}X_t=J_{u,z}1_{u\leq t}$ and then  $D_{u,z}^{N}S_t=S_t (exp\{D_{u,z}^{N}X_t\}-1)$ satisfying 
\begin{align}
\nonumber D_{s,z}^{{N}}S_t&=S_s(e^{J_{s,z}}-1)+\int_{s}^{t}\mu D_{s,z}^{N}S_udu+\int_{s}^{t}\sigma_{1}D_{s,z}^{N}S_udW_u\nonumber\\
&+\int_{s}^{t}\int_{\mathbb{R}_{0}}(e^{J_{u,z}}-1)D_{s,z}^{N}S_u\tilde{N}(du,dz).
\end{align}
Thanks to Theorem 5.6.1 in \cite{n19} and Proposition \ref{prodelta}, if there exists a random variable $u(.,.) \in Dom(\delta^N)$ such that  
\begin{align}\label{existu}
	\mathbb{E}\Big(f^{\prime}(S_T)\frac{\partial{S_T}}{\partial{S_0}}\Big)&=\mathbb{E}\Big(\int_{0}^{T}\int_{\mathbb{R}_{0}}u(t,z)(f(S_T+D_{t,z}^{{N}}S_T)-f(S_T))C_z\lambda_tdzdt\Big)\nonumber \\
 	&=\mathbb{E}\Big(\int_{0}^{T}\int_{\mathbb{R}_{0}}u(t,z)(f(S_Te^{J_{t,z}})-f(S_T))C_z\lambda_t dzdt\Big),
\end{align}
then $\Delta^N:=\frac{\partial }{\partial S_0}\mathbb{E}\Big(f(S_T)\Big) =\mathbb{E}\Big(f(S_T)\delta^N(u)\Big)$. \\
Now we calculate the delta with respect to the Poisson process in the following examples desired in \cite{huehne}.\\
{\bf Example}:
Consider the European call option with the payoff function $f(S_T)=\max(S_T-K,0)$. In fact, one can define the function $u$ of the form 
\begin{align}\label{equ27}
	u(t,z)=&\begin{cases}
		\frac{{\frac{\partial{S_T}}{\partial{S_0}}H_K(S_T)}}{\int_{0}^{T}\int_{\mathbb{R}_{0}}D_{t,z}S_TC_z\lambda_t dzdt}  \;&if D_{t,z}S_T+{S_T-K}\ge 0\\\\
		\frac{{\frac{\partial{S_T}}{\partial{S_0}}}H_K(S_T)}{\int_{0}^{T}\int_{\mathbb{R}_{0}}(K-S_T)C_z\lambda_t dzdt}  \;&if D_{t,z}S_T+{S_T-K}<0,
	\end{cases}
\end{align}
where $H_y(x)=1_{x \geq y}$ is the Heaviside function and $1_A$ is the indicator function of the set $A$. Obviously, the equality \eqref{existu} will be held for this function. Also, it is in the domain of $\delta^N$, due to the similar proof of Lemma 5.1 in \cite{Alos} for every $p \geq 2$ instead of $\frac12$, we have $\mathbb{E}\Big((\int_{0}^{T}\lambda_t dt)^{-p}\Big) < \infty$. Rewrite the definition of the function $u$ in \eqref{equ27} in the following form.
$$S_0u(t,z)=\frac{H_K(S_T)}{\int_{0}^{T}v_t\lambda_t dt}1_{S_T e^{J_{t,z}}-K \ge 0}+\frac{H_K(S_T)}{\int_{0}^{T}\lambda_t dt}\frac{S_T}{K-S_T}
1_{S_Te^{J_{t,z}}-K<0}.$$
Therefore, 
\begin{align*}
	\nonumber \Delta^{N}&=\frac{\partial}{\partial{S_0}}\mathbb{E}(f(S_T))=\mathbb{E}(f^{\prime}(S_T)\frac{\partial{S_T}}{\partial{S_0}})=\mathbb{E}\Big(f(S_T)\delta^N(u)\Big)\\
	\nonumber&=\mathbb{E}\Big(f(S_T)\frac{1}{S_0}\delta^N\Big(\frac{H_K(S_T)}{\int_{0}^{T}v_t\lambda_t dt}1_{S_T e^{J_{t,z}}-K \ge 0}+\frac{H_K(S_T)}{\int_{0}^{T}\lambda_t dt}\frac{S_T}{K-S_T}
1_{S_Te^{J_{t,z}}-K<0}\Big)\Big).
\end{align*}
According to \eqref{existu}, 
\begin{align}
	\nonumber \Delta^{N}
	&=\mathbb{E}\Big(\frac{S_TH_K(S_T)}{S_0\int_{0}^{T}v_t\lambda_tdt}\int_{0}^{T}\int_{\mathbb{R}_{0}} 1_{S_T e^{J_{t,z}}-K \ge 0} (e^{J_{t,z}}-1)C_z\lambda_t dzdt 	\Big)\notag\\		
	&+\mathbb{E}\Big(\frac{S_T H_K(S_T)}{S_0\int_{0}^{T}\lambda_tdt}\int_{0}^{T}\int_{\mathbb{R}_{0}} 1_{S_T e^{J_{t,z}}-K < 0} C_z\lambda_t dzdt 	\Big).\label{deltann}
\end{align}
{\bf In the second approach.}\\
In this part, we need the following assumption. 
\begin{Assumption}\label{asum}
For $\alpha \in (0,2)$ and some constants $c_0$ and $c$, 
\begin{equation*}
		C_. \in C^1(\mathbb{R}_0), \qquad  \vert  \frac{ \partial}{ \partial z} log C_z \vert \leq c_0 \rho(z),
\end{equation*}
	and 
\begin{equation}\label{epsillon}
		\lim_{\epsilon \rightarrow 0} \epsilon^{\alpha-2} \int_{\vert z \vert \leq \epsilon} \vert z\vert^2 C_z dz = c.	
\end{equation}
\end{Assumption}
As a result of the assumption \eqref{epsillon}, shown in \cite{song} Lemma 2.5, for any $p \geq 2$, there exist some constants $c_{0,p}$ and  $c_{1,p}$ such that 
\begin{equation}\label{epsilo}
c_{0,p} \epsilon^{p-\alpha}\leq  \int_{\vert z \vert \leq \epsilon} \vert z\vert^p C_z dz \leq c_{1,p} \epsilon^{p-\alpha}.
\end{equation}
{\bf Condition  K1:} First and second derivatives of the function $J$ with respect to $z$ is bounded, i.e., there exists some non-negative constant $\gamma$ and $c_J >0$ such that 
\begin{equation*}
\sup_{0 \leq t \leq T, z \in \mathbb{R}_0}\vert \frac{\partial J_{t,z}}{\partial z}\vert^{-1} \leq c_J \vert z \vert^{-\gamma},  \qquad    \sup_{0 \leq t \leq T, z \in \mathbb{R}_0}\vert \frac{\partial^2 J_{t,z}}{\partial z^2}\vert \leq c_J \vert z \vert^{\gamma-1}.
\end{equation*}	
In the same way as the proof of Lemma 4.1 in \cite{song}, one can show the following lemma. 
\begin{lemma}\label{lemvaroon}
	Under Assumption \rm{\ref{asum}}, for every $p \geq 2$ and $\theta \geq 2$, there exists some constant $c_p$ such that for every $t\in [0,T]$ and $\epsilon \in (0,1)$,
	\begin{equation*}
		\mathbb{E}\Big(\Big[ \int_{0 <\vert z \vert \leq \epsilon} \int_0^t\vert z\vert^\theta N(ds,dz)\Big]^{-p}\vert \mathcal{F}_t^\lambda\Big) \leq c_p \Big(\epsilon^{\theta-\alpha}\int_0^t \lambda_s ds\Big)^{-p} +\Big(\int_0^t \lambda_s ds\Big)^{-\frac{\theta p}{\alpha}}.
	\end{equation*}
\end{lemma}	
\begin{proof} 
According to \eqref{expglambda} and the proof of Lemma 4.1 in \cite{song}, we have 
\begin{align*}
\mathbb{E}\Big(\Big[\int_{0 <\vert z \vert \leq \epsilon}&\int_0^t  \vert z\vert^\theta N(ds,dz)\Big]^{-p}\vert \mathcal{F}_t^\lambda\Big) \\
&\leq \frac{1}{\Gamma(p)} \int_0^\infty  r^{p-1}exp\{\int_{\{0 < \vert z \vert \leq \epsilon\}}\int_0^t (e^{-r\vert z \vert^\theta}-1) \lambda_t C_z dt dz\} dr\\
&\leq \frac{1}{\Gamma(p)} \int_0^\infty  r^{p-1}exp\{-\int_{\{0 < \vert z \vert \leq \epsilon \wedge r^{-\frac{1}{\theta}}\}}\int_0^t c_0 r\vert z \vert^\theta \lambda_t C_z dt dz\} dr\\
& \leq c_p \Big(\epsilon^{\theta-\alpha}\int_0^t \lambda_s ds\Big)^{-p} +\Big(\int_0^t \lambda_s ds\Big)^{-\frac{\theta p}{\alpha}}.
\end{align*}
for some $c_0 >0$ that $1-e^{-x} \geq c_0 x$ as $\vert x \vert \leq 1$.
\end{proof}
Now, we calculate the delta with respect to the Poisson process using the second approach to define the Malliavin derivative. To do this, we observe that, 
based on the definition of the Malliavin derivative in this approach and \eqref{flow1}, we know 
\begin{equation}\label{dnp}
	D_{r,z}^{N_p}S_T=\frac{\partial J_{r,z}}{\partial z}S_T,
\end{equation}
satisfying the following equation for every $0\leq s \leq t$ 
\begin{align*}
 D_{s,z}^{{N_p}}S_t&=S_s\frac{\partial J_{s,z}}{\partial z}e^{J_{s,z}}+\int_{s}^{t}\mu D_{s,z}^{N_p}S_udu+\int_{s}^{t}\sigma_{1}D_{s,z}^{N_p}S_udW_u\\
 &+\int_{\mathbb{R}_{0}}\int_{s}^{t}(e^{J_{u,z}}-1)D_{s,z}^{N_p}S_u\tilde{N}(du,dz).
\end{align*}  
With a similar way to \cite{song}, set $\mathcal{A}(t,z):=\frac{1}{S_0}\Big(\frac{\partial J_{t,z}}{\partial z}\Big)^{-1} \xi(z)$ where $\xi$ is a non-negative smooth function that 
\begin{equation*}
 	\xi(z)=\vert z \vert^{3+\gamma} ~~if ~ \vert z \vert \leq \frac{1}{4}\Big(\int_0^T \mathbb{E}(\lambda_s)ds\Big)^{\frac{1}{\alpha}},   \qquad   \xi(z)=0  ~~ if ~ \vert  z \vert \geq \frac{1}{2}\Big(\int_0^T \mathbb{E}(\lambda_s)ds\Big)^{\frac{1}{\alpha}},  
\end{equation*}
and $\vert \frac{\partial}{\partial z} \xi(z)\vert \leq c_1\vert z \vert^{2+\gamma}$ and $\vert \xi(z)\vert \leq  c_1\vert z \vert^{3+\gamma}$, for some constant $c_1$. 
Then, according to Lemma \ref{lemvaroon}, under Assumption \ref{asum} and condition $K1$, one can arrive at 
\begin{equation*}
 \mathbb{E}\Big(\mathcal{N}_\xi\Big)^{-p}:=\mathbb{E}\Big(\int_{\mathbb{R}_{0}}\int_0^T \xi(z) {N}(dr,dz)\Big)^{-p} \leq  2c_p \Big(\int_0^T \mathbb{E}(\lambda_s)ds\Big)^{-\frac{(3+\gamma) p}{\alpha}},
\end{equation*}
 and for some constant $c_{Jp}$, in connection with \eqref{epsilo},
 \begin{align*}
\Vert \mathcal{A} \Vert^p_{\mathbb{V}_p} & \leq 2^{p-1}(\Vert \frac{ \partial \mathcal{A}}{\partial z}\Vert^p_{\mathbb{L}_p}+\Vert \rho\mathcal{A}\Vert^p_{\mathbb{L}_p})\\
& \leq c_{Jp} \Big[\mathbb{E}\Big(\int_{0 < \vert z \vert \leq  (\int_0^T \mathbb{E}(\lambda_s)ds)^{\frac{1}{\alpha}}}\int_0^T \vert z \vert^2 \lambda_s ds C_zdz \Big)^p\\
&+ \mathbb{E}\Big(\int_{0 < \vert z \vert \leq  (\int_0^T \mathbb{E}(\lambda_s)ds)^{\frac{1}{\alpha}}}\int_0^T \vert z \vert^{2p} \lambda_s ds C_zdz \Big)\Big] \\
& \leq   c_{1,p} c_{Jp} \Big[\Big(\int_0^T \mathbb{E}(\lambda_s)ds\Big)^{\frac{p(2-\alpha)}{\alpha}}  \mathbb{E}\Big(\int_0^T  \lambda_s ds\Big)^p + \Big(\int_0^T \mathbb{E}(\lambda_s)ds\Big)^{\frac{2p}{\alpha}} \Big] < \infty.
 \end{align*}
 Now, multiply \eqref{dnp} in $\mathcal{A}$ and get integration to derive the Poisson-Malliavin weight of the computation of delta.
 \begin{equation*}
<D^{N_p}S_T , \mathcal{A} >_N=\int_{\mathbb{R}_{0}}\int_0^T D_{r,z}^{N_p}S_T\Big(\frac{\partial J_{r,z}}{\partial z}\Big)^{-1} \xi(z){N}(dr,dz) = S_T\mathcal{N}_\xi,
\end{equation*}
and then, in connection with Propositions \ref{prod} and \ref{bypart}, 
\begin{align}
\Delta_p^{N}:=\frac{\partial}{\partial S_0}\mathbb{E}( f(S_T))&=\mathbb{E}\Big(f^{\prime}(S_T)\frac{\partial{S_T}}{\partial{S_0}}\Big)=\mathbb{E}\Big(<D^{N_p}f(S_T) , \mathcal{A} >_N \frac{1}{\mathcal{N}_\xi}\Big)\notag\\
&=\mathbb{E}\Big(f(S_T) \frac{1}{\mathcal{N}_\xi}\int_{\mathbb{R}_0}\int_0^T \frac{1}{C_z}\frac{\partial (C_.\mathcal{A})(s,z))}{\partial z}\tilde{N}(ds,dz) \Big)\notag\\
&=:\mathbb{E}\Big(f(S_T) \frac{1}{\mathcal{N}_\xi}\delta^{N_p}(\mathcal{A})\Big). \label{deltannp}
\end{align}
\begin{lemma}
Under Assumption \rm{\ref{asum}} and Condition $K1$, for every $p \geq 2$, 
\begin{equation*}
\mathbb{E}\Big(\delta^{N_p}(\mathcal{A})\Big)^p < \infty.
\end{equation*}
\end{lemma}
\begin{proof}
From Assumption \ref{asum} and Section 5.1.1 of \cite{n21}, there exist constants $c_{jp}$ and $a$ such that 
\begin{align*}
\mathbb{E}\Big(\delta^{N_p}(\mathcal{A})\Big)^p &\leq  c_{jp}\mathbb{E}\Big( \int_{0}^{T}\int_{\mathbb{R}_{0}}\frac{1}{C_z^2}[\frac{\partial  (C_.\mathcal{A})(s,z))}{\partial z}]^2 \lambda_s ds C_zdz \Big)^{\frac{p}{2}}\\
&+  c_{jp}\mathbb{E}\Big( \int_{0}^{T}\int_{\mathbb{R}_{0}}\frac{1}{C_z^p}[\frac{\partial  (C_.\mathcal{A})(s,z))}{\partial z}]^p \lambda_s ds C_zdz \Big)\\
&\leq 2^pc_{jp}\mathbb{E}\Big( \int_{0}^{T}\int_{\mathbb{R}_{0}}(\vert \frac{ \partial}{ \partial z} log C_z \vert^2 \mathcal{A}^2(s,z) \lambda_s ds C_zdz\Big)^{\frac{p}{2}}\\
&+2^pc_{jp}\mathbb{E}\Big( \int_{0}^{T}\int_{\mathbb{R}_{0}}(\vert \frac{ \partial}{ \partial z} log C_z \vert^{p} \mathcal{A}^{p}(s,z) \lambda_s ds C_zdz\Big)\\
&+ 2^pc_{jp}\mathbb{E}\Big( \int_{0}^{T}\int_{\mathbb{R}_{0}}\mathcal{A}^2(s,z) \lambda_s ds C_zdz\Big)^{\frac{p}{2}}+
2^pc_{jp}\mathbb{E}\Big( \int_{0}^{T}\int_{\mathbb{R}_{0}} \mathcal{A}^{p}(s,z) \lambda_s ds C_zdz\Big)\\
& \leq a\mathbb{E}\Big( \int_{0}^{T}\int_{\mathbb{R}_{0}}\vert z \vert^2 \lambda_s ds C_zdz\Big)^{\frac{p}{2}}+a\mathbb{E}\Big( \int_{0}^{T}\int_{\mathbb{R}_{0}}\vert z \vert^{2p} \lambda_s ds C_zdz\Big)\\
&+a\mathbb{E}\Big( \int_{0}^{T}\int_{\mathbb{R}_{0}}\vert z \vert^6\lambda_s ds C_zdz\Big)^{\frac{p}{2}}+ a\mathbb{E}\Big( \int_{0}^{T}\int_{\mathbb{R}_{0}}\vert z \vert^{3p}\lambda_s ds C_zdz\Big) < \infty.
\end{align*}
\end{proof}
\section{The convergence of the Euler scheme}\label{sec6}
	Using \cite{n17}, for any $n\in \mathbb{N}$, consider equidistant partition of the interval $[0,T]$: $t_i=t_i(n)=\frac{iT}{n} , i = 0, 1, 2, ..., n$ and define the discretizations of Wiener and Poisson processes $W^{S}$, $W$ and $N$:
	\begin{equation*}
		\Delta P_i=P(t_{i+1})-P(t_i),\;P=W^{S},\;W,\;N, i = 0, 1, 2, ..., n.
	\end{equation*}
	Discretized process of X, corresponds to the  given partition has the form
	\begin{align*}
		X_{t_j}^{n}&=logS_{t_i}=X_0+	(\mu-\frac{\sigma_{1}^{2}}{2})t_j+\sigma_{1}W_{t_j}^S+\int_{0}^{t_j}\int_{\mathbb{R}_{0}}(1-e^{J^n_{s,z}})C_z\lambda^n_sdzds+\int_{0}^{t_j}\int_{\mathbb{R}_{0}}J^n_{s,z}N^n(ds,dz)
	\end{align*}
where $\lambda_s^{n}=\lambda_{t_i}^{n}$ and $J^n_{s,z}=J^n_{t_i,z}$, for $s\in[t_i,t_{i+1}]$, and define $S_{t_j}^{n}=\exp\{X_{t_j}^{n}\}$.\\
Here, $N^n(dt,dz)$ is a poisson process independent of $N$ with stochastic intensity $\lambda^n$. So, we note that $N - N^n$ is a poisson process with intensity $\lambda-\lambda^n$.
Considering the convergence of the CIR model from \cite{cozma} and \cite{dereich2012}, for any $p\geq 1$ there exists a constant $c_1$ depending on $p$ such that 
	\begin{equation}\label{unibound}
		\sup_{s\in[0,T]}\mathbb{E}(\vert \lambda_s-\lambda_s^{n} \vert^p )\leq c_1n^{-\frac{p}{2}}, \quad \sup_{n\in\mathbb{N}}\mathbb{E}(\sup_{0 \le s \le T}\vert \lambda_s^{n}\vert^{p})<\infty, \quad \sup_{n\in\mathbb{N},s\in[0,T]}\mathbb{E}(\vert \lambda_s^{n}\vert^{-p})<\infty, \qquad p\sigma_2^2 < 2\kappa \Theta.
	\end{equation}
\begin{remark}\label{sappinverse}
According to Lemma \ref{pp0} and \eqref{sinverse}, with a similar way, one can show that for $0\le p \le \frac14p_0$, approximating process and its inversion have uniformly bounded moments, 
\begin{equation*}
\sup_{t\in[0,T]}\mathbb{E}(\vert S_t^n \vert^{p})<\infty,\qquad if ~ 2p \sigma_2^2 < \kappa, \qquad ~\qquad	\sup_{t\in[0,T]}\mathbb{E}(\vert S_t^n \vert^{-4})<\infty,\qquad if ~ (16u_1 \vee 1)\sigma_2^2 < \kappa.
\end{equation*}
\end{remark}
The following two inequalities are required to establish the convergance rate. For every $x,y\in\mathbb{R}$ and $p\in\mathbb{N}$,
\begin{equation}\label{e}
	\vert e^{x}-e^{y}\vert\leq(e^{x}+e^{y})\vert x-y\vert,
\end{equation}
\begin{equation}\label{2p}
	(x+y)^{2p}\leq 2^{2p-1}(x^{2p}+y^{2p}).
\end{equation}
Let $A_{u,s}:= exp\{-\int_{u}^{s}(\frac{\kappa}{2}+\frac{C_{\sigma}}{\lambda_r})dr\}-1$ and $A^n_{u,s}:= exp\{-\int_{u}^{s}(\frac{\kappa}{2}+\frac{C_{\sigma}}{\lambda^n_r})dr\}-1$. From \eqref{e}, for every $x,y >0$, 
$$\vert e^{-x}-e^{-y} \vert =\vert \frac{e^x-e^y}{e^xe^y}\vert \le \vert \frac{1}{e^x}+\frac{1}{e^y}\vert \vert x-y\vert \leq 2\vert x-y\vert.$$ Then for every $0 \leq t \leq s \leq T$, and $4p \sigma_2^2 < \kappa \Theta$
\begin{align}\label{aats}
\mathbb{E}(\vert A_{t,s}-A_{t,s}^n \vert^{2p}) & \leq (2C_{\sigma})^{2p}\mathbb{E}\Big( \int_{t}^{s}(\frac{1}{\lambda_r}-\frac{1}{\lambda^n_r})dr \Big)^{2p} \notag\\ 
&\leq (2TC_{\sigma})^{2p} \int_{t}^{s}\mathbb{E}^\frac12({\lambda_r}-{\lambda^n_r})^{4p}\Big[\mathbb{E}(\frac{1}{\lambda_r^{8p}})+\mathbb{E}(\frac{1}{(\lambda_r^n)^{8p}})\Big]^\frac12 dr \notag\\
&\leq  (2TC_{\sigma})^{2p} (\triangle t)^p \int_{t}^{s}\Big[\mathbb{E}(\frac{1}{\lambda_r^{8p}})+\mathbb{E}(\frac{1}{(\lambda_r^n)^{8p}})\Big]^\frac12 dr.
\end{align}
We also define $Z_T^{n}=\delta^W(\frac{h^{n}}{\mathcal{B}^n_T})$
 in which for $u\in[t_i,t_{i+1}]$ 
 $$h^n(u)=\frac{1}{v_{t_i} }(\frac{\kappa}{2}+\frac{C_\sigma}{\lambda^n_{t_i}}), \quad ~and~\quad 
 \mathcal{B}_T^n:=- \int_0^T\Big[\sigma_{2}\sqrt{\lambda^n_{s}}(1-exp\{-\int_{0}^{s}(\frac{\kappa}{2}+\frac{C_{\sigma}}{\lambda^n_r})dr\})\Big]ds.$$  
\begin{remark}\label{hn}
When $k\Theta > 2 \sigma_2^2$, there exists some constant $E_0$ that 
\begin{equation*}
 \mathbb{E}\Big(\int_{0}^{T} (h^{n}(u))^4 du\Big) \leq 2^4\epsilon_0^{-4} \Big(T\frac{\kappa^4}{2^4}+C_\sigma^4  \mathbb{E}\Big(\int_{0}^{T} \frac{1}{(\lambda^n_{s})^4} ds\Big)\Big) < E_0 < \infty
\end{equation*}
\end{remark}
\begin{remark}\label{bn}
Obviously, for every $p \geq 2$, we know $\mathcal{B}_T^n \in L^{p}$ and in a same way of Theorem \ref{th1}, when $2\kappa \Theta > 3\sigma_2^2$,  we result $(\mathcal{B}_T^n)^{-1} \in L^{p}$.
\end{remark}
To achieve the convergence rate, we need to assume that the function $J$ is globally Lipschitz on the first parameter, time, i.e., for every $t,s \in [0, T]$, there exists some constant $c_5$ that
\begin{equation}\label{lipj}
\sup_{z \in \mathbb{R}} \vert J_{t,z} - J_{s,z}\vert \leq c_5 \vert t-s \vert.
\end{equation}
This condition can be also considered in the following weaker assumption. There exists some function $\varrho$ that 
\begin{equation*} 
 \int_{\mathbb{R}_{0}}e^{pJ_{t,z}}\varrho(z) C_zdz  < \infty, \qquad  \int_{\mathbb{R}_{0}} \varrho(z)C_z dz < \infty,
\end{equation*}
and for every $t,s \in [0, T]$, 
\begin{equation*} 
\vert J_{t,z} - J_{s,z}\vert \leq \varrho(z) \vert t-s \vert.
\end{equation*}
However, we will prove the following lemmas with assumption \ref{lipj} and the same proof is valid for the latter. \\
We first prove the convergence of $X_t^n$ and $Z_t^n$ in the following lemmas.
\begin{lemma}\label{lemma1}
There exists some constant $c_2>0$ such that
\begin{equation}\label{equboundX}
\mathbb{E}(X_T-X_T^{n})^{2p}\leq c_2n^{-\frac{p}{2}}, \qquad  \qquad \qquad  \mathbb{E}(S_T-S_T^{n})^{2p}\leq c_2n^{-\frac{p}{2}}, \quad 4p\sigma_2^2 < \kappa.
\end{equation}
\end{lemma}
\begin{lemma}\label{lemma2}
When $\kappa \Theta > 64\sigma_2^2$, there exists some constant $c'_2>0$ such that
\begin{equation}\label{equboundZ}
\mathbb{E}(Z_T-Z_T^{n})^{2}\leq c'_2n^{-1}.
\end{equation}
\end{lemma}
\begin{lemma}\label{lemma3}
Under condition {\bf{H2}}, if  $16(u_1 \vee p)\sigma_2^2 < \kappa$, then we have the following upper bound: there exists a constant $c_F$ such that
\begin{equation*}
\mathbb{E}\Big\vert\frac{F(S_T)}{S_T}-\frac{F(S_T^{n})}{S_T^{n}}\Big\vert^{2}<c_F\times n^{-\frac12}
\end{equation*}
\end{lemma}
Building on the previous explanations and the stated lemmas, we now prove the main results of this section, the convergence rate of the discretized option price and the delta.
\begin{theorem}\label{theo1}
	Let conditions {\bf{H1}} and {\bf{H2}} hold. There exists a constant $c>0$ which is not dependent on J such that
\begin{equation*}
	\vert\mathbb{E}f(S_T)-\mathbb{E}\Big(\frac{F(S_T^{n})}{S_T^{n}}(1+{Z_T^{n}})\Big)\vert\leq c n^{-\frac{1}{4}}.	
\end{equation*}
\end{theorem}
As a result, one can easily prove the convergence of approximation of delta by discretized process $S_t^n$ to the true delta presented in \eqref{deltann} and \eqref{deltannp}. 
\begin{theorem}\label{theo2}
	Let conditions {\bf{H1}} and {\bf{H2}} hold. There exists a constant $c>0$ which is not dependent on J such that
\begin{equation*}
	\mathbb{E}\Big(\Delta^N - \Delta_n^N\Big)^2 \leq c n^{-\frac{1}{2}},	\qquad  \mathbb{E}\Big(\Delta_p^N - \Delta_{n,p}^N\Big)^2\leq c n^{-\frac{1}{2}},
\end{equation*}
where $\Delta_n^N=\Delta^N\Big\vert_{S_T=S_T^n}$ and $ \Delta_{n,p}^N= \Delta_{p}^N\Big\vert_{S_T=S_T^n}$.
\end{theorem}
{\bf{Proof of Lemma \ref{lemma1}.}} Let's start with \eqref{equboundX}. Using \eqref{2p}
	\begin{align}
		\mathbb{E}\Big( X_T-X_T^{n}\Big)^{2p}&=\mathbb{E}\Big(\int_{0}^{T}\int_{\mathbb{R}_{0}}(1-e^{J_{s,z}})C_z\lambda_sdzds+\int_{0}^{T}\int_{\mathbb{R}_{0}}J_{s,z}N(ds,dz)\nonumber\\
		&-[\int_{0}^{T}\int_{\mathbb{R}_{0}}(1-e^{J_{s,z}^{n}})C_z\lambda_s^{n}dzds+\int_{0}^{T}\int_{\mathbb{R}_{0}}J_{s,z}^{n}N^n(ds,dz)]
		\Big)^{2p}\nonumber\\
		&\leq 4^{2p-1}\mathbb{E}\Big( \int_{0}^{T}\int_{\mathbb{R}_{0}}C_z(\lambda_s-\lambda_s^{n})dzds\Big)^{2p}+4^{2p-1}\mathbb{E}\Big(\int_{0}^{T}\int_{\mathbb{R}_{0}}C_z(e^{J_{s,z}^{n}}\lambda_s^{n}-e^{J_{s,z}}\lambda_s)dzds\Big)^{2p}\nonumber\\
		&+4^{2p-1}\mathbb{E}\Big(\int_{0}^{T}\int_{\mathbb{R}_{0}}J_{s,z}N(ds,dz)-\int_{0}^{T}\int_{\mathbb{R}_{0}}J_{s,z}^{n}N^n(ds,dz)\Big)^{2p}\nonumber\\
		&\leq 4^{2p-1}\Big(T^2\sup_{0 \leq s \leq T}\mathbb{E}\Big(\lambda_s-\lambda_s^{n})^{2p}\Big)+I_1+I_2\Big).\label{xxn}
\end{align}	
From Holder inequality, \eqref{e}, \eqref{lipj} and then \eqref{unibound} there exists some constant $c_6$ that  
\begin{align*}
I_1 &=\mathbb{E}\Big(\int_{0}^{T}\int_{\mathbb{R}_{0}}C_z(e^{J_{s,z}}\lambda_s-e^{J_{s,z}^{n}}\lambda_s^{n})dzds\Big)^{2p} \leq 2^{2p-1}\mathbb{E}\Big(\int_{0}^{T}\int_{\mathbb{R}_{0}}C_z(e^{J_{s,z}}-e^{J_{s,z}^{n}})\lambda_s^{n}dzds\Big)^{2p}\\
&+2^{2p-1}\mathbb{E}\Big(\int_{0}^{T}\int_{\mathbb{R}_{0}}C_ze^{J_{s,z}}(\lambda_s-\lambda_s^{n})dzds\Big)^{2p}\\
& \leq 2^{2p-1} \mathbb{E}\Big(\Big[\int_{0}^{T}\int_{\mathbb{R}_{0}}C_z(e^{J_{s,z}}-e^{J_{s,z}^{n}})^{2p}\lambda_s^n dzds \Big]\Big[ \int_0^T \lambda_s^{n}ds\Big]^{2p-1}\Big)\\
&+2^{2p-1}\mathbb{E}\Big(\Big[\int_{0}^{T}\int_{\mathbb{R}_{0}}C_ze^{2pJ_{s,z}} (\lambda_s-\lambda_s^{n})dzds\Big]\Big[\int_0^T (\lambda_s-\lambda_s^{n}) ds\Big]^{2p-1}\Big)\\
&\leq 2^{2p-1}\mathbb{E}\Big(\Big[\int_{0}^{T}\int_{\mathbb{R}_{0}}C_z(e^{J_{s,z}}+e^{J_{s,z}^{n}})^{2p} \vert J_{s,z}-J_{s,z}^n\vert^{2p}\lambda^n_s dzds \Big] \Big[ \int_0^T (\lambda_s^{n}) ds\Big]^{2p-1}\Big)+ (2T)^{2p-1}u_{2p}\mathbb{E}\Big[\sup_{0 \leq s \leq T}(\lambda_s-\lambda_s^{n})^{2p} ds\Big]\\
&\leq (4Tc_5)^{2p} (\triangle t)^{2p} u_2 \mathbb{E}\Big(\int_0^T (\lambda_s^{n})^{2p} ds\Big)+(2T)^{2p-1}u_{2p}\sup_{0 \leq s \leq T}\mathbb{E}\Big((\lambda_s-\lambda_s^{n})^{2p} ds\Big)\\
& \leq c_6 (\triangle t)^{p}.
\end{align*}
Also, from \eqref{supsup} and \eqref{unibound} there exists some constant $c_7$ that  
\begin{align*}		
I_2 &=\mathbb{E}\Big(\int_{0}^{T}\int_{\mathbb{R}_{0}}J_{s,z}N(ds,dz)-\int_{0}^{T}\int_{\mathbb{R}_{0}}J_{s,z}^{n}N^n(ds,dz)\Big)^{2p}\\
&\leq 2^{2p-1}\mathbb{E}\Big(\int_{0}^{T}\int_{\mathbb{R}_{0}}(J_{s,z}-J_{s,z}^{n})^{2p}C_z\lambda_sdsdz\Big)+2^{2p-1}\mathbb{E}\Big(\int_{0}^{T}\int_{\mathbb{R}_{0}}(J_{s,z}-J_{s,z}^{n})^{2}C_z\lambda_sdsdz\Big)^p\\
&+2^{2p-1}\mathbb{E}\Big(\int_{0}^{T}\int_{\mathbb{R}_{0}}J^n_{s,z}(N(ds,dz)-N^n(ds,dz))\Big)^{2p} \\
		&\leq (2c_5)^{2p} (\triangle t)^{2p}\Big[\mathbb{E}\Big( \int_{0}^{T} \lambda_sds\Big)+\mathbb{E}\Big( \int_{0}^{T} \lambda_sds\Big)^p\Big]+2^{2p-1}\mathbb{E}\Big(\int_{0}^{T}\int_{\mathbb{R}_{0}}(J^n_{s,z})^{2p}C_z(\lambda_s-\lambda_s^{n})dzds \Big)\\
		&+2^{2p-1}\mathbb{E}\Big(\int_{0}^{T}\int_{\mathbb{R}_{0}}(J^n_{s,z})^{2}C_z(\lambda_s-\lambda_s^{n})dzds \Big)^p\\
		&\leq(2c_5)^{2p} (\triangle t)^{2p}\Big[\mathbb{E}\Big( \int_{0}^{T} \lambda_sds\Big)+\mathbb{E}\Big( \int_{0}^{T} \lambda_sds\Big)^p\Big]+(2T)^{2p}u_{2p}\Big[\sup_{0\leq s \leq T}\mathbb{E}\Big(\lambda_s-\lambda_s^{n}\Big)+\sup_{0\leq s \leq T}\mathbb{E}\Big(\lambda_s-\lambda_s^{n}\Big)^p\Big]\\
		&\leq c_7  (\triangle t)^{\frac{p}{2}},
\end{align*}
where we used the fact that $\mathbb{E}\Big(\lambda_s-\lambda_s^{n}\Big)=0$ in the last inequality. Substitute the bounds of $I_1$ and $I_2$ into equation \eqref{xxn} to complete the proof. \\
{\bf{Proof of Lemma \ref{lemma2}.}}  
 Let 
 \begin{align*}
J_0 := \mathbb{E}\left(\left|\langle\frac{D^W \mathcal{B}_T}{\mathcal{B}_T^2}-\frac{D^W \mathcal{B}_T^n}{\left(\mathcal{B}_T^n\right)^2}, h\rangle_{L^2([0, T])}-\langle\frac{D^W \mathcal{B}_T^n}{\left(\mathcal{B}_T^n\right)^2}, h^n-h\rangle_{L^2([0, T])}\right|^2\right).
\end{align*} 
Using the equation \eqref{deltaboundd}, the second part of condition {\bf{H1}} and Remark \ref{hn}, we have
	\begin{align*}
		\mathbb{E}(Z_T-Z_T^{n})^2&=\mathbb{E}\Big[\frac{1}{\mathcal{B}_T}\int_{0}^{T}h(u)dW_u-\frac{1}{\mathcal{B}^n_T}\int_{0}^{T}h^{n}(u)dW_u\Big]^2+J_0\\
		&\leq 2\mathbb{E}\Big[\frac{1}{\mathcal{B}_T} \int_{0}^{T}(h(u)-h^{n}(u))dW_u\Big]^2+2\mathbb{E}\Big[(\frac{1}{\mathcal{B}_T} -\frac{1}{\mathcal{B}^n_T})\int_{0}^{T} h^{n}(u))dW_u\Big]^2+J_0\\
		&\leq 2\mathbb{E}\Big[\frac{1}{\mathcal{B}_T}\Big]^2\int_{0}^{T}\mathbb{E}(h(u)-h^{n}(u))^{2}du+ 4T \mathbb{E}^\frac12\Big[(\frac{1}{\mathcal{B}_T} -\frac{1}{\mathcal{B}^n_T})^4 \Big] \mathbb{E}^\frac12\Big(\int_{0}^{T} (h^{n}(u))^4 du\Big)+J_0 \\
		&=: 2\mathbb{E}\Big[\frac{1}{\mathcal{B}_T}\Big]^2J_1+ 4TE_0^\frac12 J_2+J_0.
\end{align*}
To find the upper bound $J_1$, let $\nu^n_u=\nu^n_{t_i}$, for every $u \in [t_i, t_{i+1}]$ and $E_p:=\mathbb{E}\Big(\frac{1}{\nu_u}-\frac{1}{\nu_u^{n}}\Big)^{p}$, for every $p\geq 2$. Using \eqref{e} and \eqref{assum1} we deduce 
\begin{align*}
E_p &\leq  \epsilon_0^{-2p}\mathbb{E}\Big(\nu_u-\nu_u^{n}\Big)^{p}
 \leq \epsilon_0^{-2p} \mathbb{E}\int_{\mathbb{R}_0}\Big(e^{J_{s,z}}-e^{J_{s,z}^n}\Big)^{p}C_z dz
\leq \epsilon_0^{-2p} \mathbb{E}\int_{\mathbb{R}_0}\Big(e^{pJ_{s,z}}+e^{pJ_{s,z}^n}\Big)\vert J_{s,z}-J_{s,z}^n\vert^{p}C_z dz \leq 2\epsilon_0^{-2p}u_p(\triangle t)^{p}. 
\end{align*}
Now, from above inequality and \eqref{unibound} we derive
\begin{align*}
	J_1= &\int_{0}^{T}\mathbb{E}\Big( \frac{1}{\nu_u}(\frac{\kappa}{2}+\frac{C_{\sigma}}{\lambda_u})-\frac{1}{\nu^n_u}(\frac{\kappa}{2}+\frac{C_{\sigma}}{\lambda_u^{n}})\Big)^{2}du\\
		&\le 2\int_0^T \mathbb{E}\Big( \frac{\kappa^2}{4}(\frac{1}{\nu_u}-\frac{1}{\nu^n_u})\Big)^2du+2C_{\sigma}^2 \int_{0}^{T}\mathbb{E}(\frac{1}{\nu_u\lambda_u}-\frac{1}{\nu_u^n\lambda_u^{n}}\Big)^{2}du\\
		&\le 2\int_0^T \mathbb{E}\Big( \frac{\kappa^2}{4}(\frac{1}{\nu_u}-\frac{1}{\nu^n_u})\Big)^2du+4C_{\sigma}^2 \int_{0}^{T}\mathbb{E}\Big(\frac{1}{\lambda_u^2}\Big( \frac{1}{\nu_u^n}-\frac{1}{\nu_u}\Big)^{2}\Big)du+4C_{\sigma}^2 \int_{0}^{T}\mathbb{E}\Big((\frac{1}{\lambda^n_u})^2\Big( \frac{1}{\lambda_u^n}-\frac{1}{\lambda_u}\Big)^{2}\Big)du\\
		&\leq \frac{\kappa^2}{\epsilon_0^4}u_2(\triangle t)^{2} +4\sqrt{2u_4}\frac{C_{\sigma}^2}{\epsilon_0^4}(\triangle t)^{2} \int_{0}^{T}\mathbb{E}^\frac12\Big(\frac{1}{\lambda_u^4}\Big)du + 4C_{\sigma}^2 \int_{0}^{T}\mathbb{E}^\frac12\Big(\frac{1}{(\lambda^n_u)^8\lambda_u^4}\Big)\mathbb{E}^\frac12\Big(\lambda^n_u-\lambda_u\Big)^4 du\\
& \leq c_8 (\triangle t),
     \end{align*}
where $c_8$ is a positive constant when $\kappa \Theta > 8\sigma_2^2$.\\
To bound  $J_2$, from the equation \eqref{aats}, Holder inequality, above inequality, and Remark \ref{bn} there exists some constants $c_9$ and $c_{10}$ that  
\begin{align*}
J_2^2 =\mathbb{E}\Big[(\frac{1}{\mathcal{B}_T} -\frac{1}{\mathcal{B}^n_T})^4 \Big]  &\leq \frac12 \mathbb{E}\Big(\frac{1}{\mathcal{B}_T^{16}}+\frac{1}{(\mathcal{B}^n_T)^{16}}\Big)\mathbb{E}\Big[(\mathcal{B}_T -\mathcal{B}^n_T)^8\Big]\\
& \leq c_8 \sigma_{2}^8 \mathbb{E}\Big( \int_0^T\Big[\sqrt{\lambda_{s}}A_{0,s} -\sqrt{\lambda^n_{s}}A_{0,s}^n\Big]ds\Big)^8\\
& \leq 2^8 c_9 T^8\sigma_{2}^8 \Big[\mathbb{E}\Big( \int_0^T A_{0,s}^8\Big[\sqrt{\lambda_{s}} -\sqrt{\lambda^n_{s}}\Big]^8 ds\Big)+\mathbb{E}\Big( \int_0^T\Big[\sqrt{\lambda^n_{s}}A_{0,s} -\sqrt{\lambda^n_{s}}A_{0,s}^n\Big]^8 ds\Big)\Big]\\
&\leq 2^8 c_9 T^8\sigma_{2}^8 \Big[2^{16}\mathbb{E}\Big( \int_0^T \Big[\sqrt{\lambda_{s}} -\sqrt{\lambda^n_{s}}\Big]^8 ds\Big)+2^8C_{\sigma}^8\mathbb{E}\Big( \int_0^T\Big[(\lambda^n_{s})^4  \Big\vert \int_{0}^{s}(\frac{1}{\lambda_r}-\frac{1}{\lambda^n_r})dr \Big\vert^8 ds\Big)\Big]\\
& \leq c_{10} (\triangle t)^2
\end{align*}
where we used \eqref{unibound} in the last inequalty when $2\sigma_2^2 < \kappa \Theta$. Finaly, we  find an upper bound for $J_0$. To this end, we first know that 
\begin{align*}
D_t^W \mathcal{B}_T^n=\sigma_2^2 \int_t^T\left[\frac{1}{2} A_{0, s}^n\left(1+A_{t, s}^n\right)-C_0 \sqrt{\lambda_s^n}\left(2+A_{0, s}^n\right) \int_t^s \frac{1}{\lambda_r^n \sqrt{\lambda_r^n}}\left(1+A_{t, r}^n\right) d r\right] d s, 
\end{align*}
and 
\begin{align*}
D_t^W \mathcal{B}_T=\sigma_2^2 \int_t^T\left[\frac{1}{2} A_{0, s}\left(1+A_{t, s}\right)-C_0 \sqrt{\lambda_s}\left(2+A_{0, s}\right) \int_t^s \frac{1}{\lambda_r \sqrt{\lambda_r}}\left(1+A_{t, r}\right) d r\right] ds, 
\end{align*}
have uniformly bounded $p$-moments for every $p \geq 2$ which deduces the processes $Z_t$ and $Z_t^n$ have also uniformly bounded $p$-moments.\\
From Holder inequality, and the fact that $2xy \leq x^2+y^2$, we derive 
\begin{align*}
J_0 & \leqslant 2^2 \mathbb{E}\left(\|h\|_{L^2}^2 \frac{1}{\mathcal{B}_T^4}\left\|D^W \mathcal{B}_T-D^W \mathcal{B}_T^n\right\|_{L^2}^2\right)+2^2 \mathbb{E}\left(\|h\|_{L^2}^2\left|\frac{1}{\mathcal{B}_T^2}-\frac{1}{\left(\mathcal{B}_T^n\right)^2}\right|^2\left\|D^W \mathcal{B}_T^n\right\|_2^2\right) \\
& +2 \mathbb{E}\left(\frac{1}{\left(\mathcal{B}_T^n\right)^4}\left\|D^W \mathcal{B}_T^n\right\|_{L^2}^2\left\|h^n-h\right\|_{L^2}^2\right) \\
& \leq 4 T^2 \mathbb{E}^\frac14 \left(\frac{1}{\mathcal{B}_T^{16}}\right) \mathbb{E}^\frac14\left(\int_0^T h^8(s) d s\right) \mathbb{E}^{\frac12}\left(\left\|\left(D^W \mathcal{B}_T-D^W \mathcal{B}_T\right)^2\right\|_{L^2}^2\right) \\
&+ \frac{4}{\sqrt{2}} \mathbb{E}^{\frac12}\left(\left|\frac{1}{\mathcal{B}_T^2}-\frac{1}{(\mathcal{B}_T^n)^2}\right|^4\right) \mathbb{E}\left(\|h\|_{L_2}^8+\left\|D^W \mathcal{B}_T^n\right\|_{L^2}^8\right)+ \frac{2}{\sqrt{2}} \mathbb{E}^{1 / 2}\left(\frac{1}{\left(\mathcal{B}_T^n\right)^{16}}+\left\|D^W \mathcal{B}_T^n\right\|_{L^2}^8\right) \mathbb{E}^{1 / 2}\left(\left\|h^n-h\right\|_{L^2}^4\right). \\
\end{align*}
Then there exist some constant $c_{11}$ that 
\begin{align}
J_0 \leq c_{11} \Big[  \mathbb{E}^{\frac12}\left(\left\|\left(D^W \mathcal{B}_T-D^W \mathcal{B}_T\right)^2\right\|_{L^2}^2\right)+ (\triangle t)^2\Big]. \label{jj0}
\end{align}
Now, we should identify the upper bound on the right-hand side of the above inequality. 
\begin{align}
\left|D^W_t \mathcal{B}_T-D^W \mathcal{B}_T\right|^4 & \leqslant T \sigma_2^8\left[\int_t^T\left|A_{0,s}\left(1+A_{t, s}\right)-A_{0, s}^n\left(1+A_{t, s}^n\right)\right|_{s_s^4}^4+C_\sigma^4 \int_t^T\left|J_5(s)\right|^4 ds\right] \notag\\
& \leq T \sigma_2^8\left[ 6 \left|A_{0, s}-A_{0, s}^n\right|^4+ 4  \left|A_{t, s}-A_{t, s}^n\right|^4+C_\sigma^4 \int_t^T\left|J_5(s)\right|^4 ds\right], \label{c11}
\end{align}
where 
\begin{align*}
 J_5(s):=\sqrt{\lambda_s}\left(2+A_{0, s}\right) \int_t^s \frac{1}{\lambda_r}\left(1+A_{t, r}\right) dr-\sqrt{\lambda_s^n}\left(2+A_{0, s}^n\right) \int_t^s \frac{1}{\lambda_r^n \sqrt{\lambda_r^n}}\left(1+A_{t, r}^n\right) dr. 
\end{align*}
Define 
\begin{align*}
 J_6(s)=\frac{\left(1+A_{t,r}\right)}{\lambda_r \sqrt{\lambda_r}}-\frac{\left(1+A_{t,r}^n\right)}{\lambda_r^n \sqrt{\lambda_r^n}}, \qquad
 J_7(s)=\sqrt{\lambda_s}\left(2+A_{0,s}\right)-\sqrt{\lambda_s^n}\left(2+A_{0,s}^n\right),
\end{align*}
and apply Holder inequality to result that there exist some constants $c_{12}$ and $c_{13}$ 
\begin{align}
\mathbb{E}\left(\left|J_5{(s)}\right|^4 \right) & \leqslant 3^4 T \mathbb{E}^{\frac12}\left(\lambda_s^4\right) \mathbb{E}^{\frac12}\left(\int_t^s\left|J_6(r)\right|^8 d r\right)+T2^4 \mathbb{E}^{\frac12}\left(\left|J_7(s)\right|^8\right) \mathbb{E}^{\frac12}\left(\int_t^5 \frac{1}{\left(\lambda_r^n\right)^{12}} dr\right)\notag\\
& \leqslant c_{12}3^4 T \left(\int_t^s \left[2 \mathbb{E}\left(\frac{\left(A_{t, r}-A_{t,r}^n\right)^{16}}{\lambda_r^{12}}\right)+4 \mathbb{E}\left(\frac{\left(\lambda_r^n \sqrt{\lambda_r^n}-\lambda_r \sqrt{\lambda_r}\right)^8}{\lambda_r^{12}\left(\lambda_r^n\right)^{12}}\right)\right]dr\right)^\frac12\notag\\
&+ c_{12}2^2 T \Big[6 \mathbb{E}^{\frac12}\Big(\vert\sqrt{\lambda_s^n}-\sqrt{\lambda_s}\vert^8\Big)+2 \mathbb{E}^{\frac12}\Big(\vert A_{0, s}-A_{0, s}^n\vert^{16}\Big) \mathbb{E}^{\frac12}(\lambda_s^n)^8\Big]^{\frac12}\notag\\
& \leqslant c_{12}3^4 T \Big(\int_t^s \Big[2 \mathbb{E}^{\frac12}\left(A_{t, r}-A_{t,r}^n\right)^{32}\mathbb{E}^{\frac12}\left(\lambda_r^{-24}\right)+4 \mathbb{E}^{\frac12}\left(\vert \sqrt{\lambda_r^n}-\sqrt{\lambda_r}\vert^{16}\right)\left[\mathbb{E}(\lambda_r^n)^{-16}+\mathbb{E}(\lambda_r)^{-48}\right]^\frac12\notag\\
&  +4 \mathbb{E}^{\frac12}\left(\vert {\lambda_r^n}-{\lambda_r}\vert^{16}\right)\left[  \mathbb{E}(\lambda_r^n)^{-32}+\mathbb{E}(\lambda_r)^{-48}\right]^\frac12 \Big]\Big)^\frac12\notag\\
&+ c_{12}2^2 T \Big[6 \mathbb{E}^{\frac12}\Big(\vert\sqrt{\lambda_s^n}-\sqrt{\lambda_s}\vert^8\Big)+2 \mathbb{E}^{\frac12}\Big(\left|A_{0, s}-A_{0, s}^n\right|^{16}\Big) \mathbb{E}^{\frac12}(\lambda_s^n)^8\Big]^{\frac12}\notag\\
& \leq c_{13} (\triangle t)^4. \label{c13}
\end{align}
Substitute \eqref{c13} into \eqref{c11} and then in \eqref{jj0} to result $ J_0 \leq c_{14} \Big[(\triangle t)^4+ (\triangle t)^2\Big]$, for some constant $c_{14}$.
The upper bounds of $J_0$, $J_1$ and $J_2$ complete the proof.\\
{\bf{Proof of Lemma \ref{lemma3}.}}
We can write
\begin{align}
	\mathbb{E}\Big(\Big\vert\frac{F(S_T)}{S_T}-\frac{F(S_T^{n})}{S_T^{n}}\Big\vert^{2}\Big)	&\leq2\mathbb{E}\Big(\Big\vert\frac{F(S_T)}{S_T}-\frac{F(S_T)}{S_T^{n}}\Big\vert^{2}\Big)+2\mathbb{E}\Big(\Big\vert\frac{F(S_T^{n})}{S_T^{n}}-\frac{F(S_T)}{S_T^{n}}\Big\vert^{2}\Big):=2I_3+2I_4.\label{equ5}
\end{align}
We estimate the right-hand side of the equation \eqref{equ5} term by term. To bound $I_3$, from \eqref{e} we have 
\begin{align*}
I_3
\leq 2\Big[\mathbb{E}^\frac12\Big((\frac{1}{S_T})^4\Big) + \mathbb{E}^\frac12\Big((\frac{1}{S^n_T})^4\Big)\Big]\mathbb{E}^\frac14\Big(F(S_T))^8\Big)\mathbb{E}^\frac14\Big((\vert X_T-X_T^n \vert)^8\Big)
\end{align*}
Applying condition {\bf{H2}}, since $f$ has a polynomial growth of order $\frac{p_0}{32}$, from Lemma \ref{pp0} we conclude that  $\mathbb{E}\Big(F(S_T)\Big)^{4}<\infty$. Also, Remark \ref{sappinverse} and Lemma \ref{lemma1} deduce
\begin{equation}\label{equ7}
I_3\leq C' n^{-\frac12}.
\end{equation}
Now to bound $I_4$, according to Remark \ref{sappinverse} and condition {\bf H2} we conclude
\begin{align*}
	I_4&\leq\Big[\mathbb{E}(F(S_T)-F(S_T^{n}))^{4}\Big]^{\frac{1}{2}}\Big[\mathbb{E}(S_T^{n})^{-4}\Big]^{\frac{1}{2}}\\
	 &\leq C'\Big[\mathbb{E}(F(S_T)-F(S_T^{n}))^{4}\Big]^{\frac{1}{2}}=C\Big[\Big(\int_{S_T\land S_T^{n}}^{S_T\lor S_T^{n}}f(x)dx\Big)^{4}\Big]^{\frac{1}{2}}\\
	&\leq C' (C_f)^{2}\big[\mathbb{E}\Big(\vert S_T-S_T^{n}\vert^{4}(1+S_T^{p}+(S_T^{n})^{p})^{4}\Big)\big]^{\frac{1}{2}}\\
	&\leq c_{15}\big[\mathbb{E}\Big(\vert S_T-S_T^{n}\vert^{8}\Big)\mathbb{E}\Big(1+S_T^{p}+(S_T^{n})^{p}\Big)^{8}\big]^{\frac{1}{4}},
\end{align*}
where $p \leq \frac{p_0}{23}$ and $C', C_f, c_{15}$ are constants.
According to Lemma \ref{pp0} and Remark \ref{sappinverse}, we have that
\begin{equation*}
	\sup_{n\in\mathbb{N}}\mathbb{E}\Big(1+S_T^{p}+(S_T^{n})^{p}\Big)^{8}<\infty.
\end{equation*}
Therefore, from Lemma \ref{lemma1}, and \eqref{e} for some constant $C_F$ we result 
\begin{equation}\label{equ8}
	I_4\leq C_F\big(\mathbb{E}\vert X_T-X_T^{n}\vert^{16}\big)^{\frac{1}{4}} \leq C_F c_2^\frac14 n^{-1}
\end{equation}
Relations \eqref{equ7} and \eqref{equ8} complete the proof.\\
{\bf{Proof of theorem \ref{theo1}.}} By Theorem \ref{thm1} we have
\begin{align*}
		\Big\vert\mathbb{E}(f(S_T))-\mathbb{E}\Big(\frac{F(S_T^{n})}{S_T^{n}}(1+\frac{Z_T^{n}}{T})\Big)\Big\vert& =\mathbb{E}\Big\vert\Big(\frac{F(S_T)}{S_T}(1+\frac{Z_T}{T})\Big)-\Big(\frac{F(S_T^{n})}{S_T^{n}}(1+\frac{Z_T^{n}}{T})\Big)\Big\vert\\
		&\leq
		\frac{1}{T}\mathbb{E}\Big\vert\frac{F(S_T)}{S_T}(Z_T-Z_T^{n})\Big\vert+\mathbb{E}\Big\vert\Big(1+\frac{Z_T^{n}}{T}\Big)\Big(\frac{F(S_T)}{S_T}-\frac{F(S_T^{n})}{S_T^{n}}\Big)\Big\vert\\
		&\leq \frac{1}{T}\Big[\mathbb{E}\Big(\frac{F(S_T)}{S_T}\Big)^{2}\mathbb{E}(Z_T-Z_T^{n})^{2}\Big]^{\frac{1}{2}}+\Big[\mathbb{E}\Big(\frac{F(S_T)}{S_T}-\frac{F(S_T^{n})}{S_T^{n}}\Big)^{2}\mathbb{E}\Big(1+\frac{Z_T^{n}}{T}\Big)^{2}\Big]^{\frac{1}{2}}.
\end{align*}
The proof follows from Lemmas \eqref{lemma3} and \eqref{lemma2}.

\section{Numerical Example}\label{sec7}
In this section, we perform two examples of error analysis in normal and Kou models. We also show the delta of the European call option in the normal case with two approaches. In Kou model, presenting the delta is also similar in that we omit it.  
\subsection{Pricing and convergence rate}
Consider the SDE \eqref{equ1} wit the parameters $\sigma_1=0.40$, $\sigma_2=0.10$, $\Theta=2$, $\kappa=0.32$, $\mu=1$, $T=1$, $\lambda_0=0.10$, $S_0=5$, the number of simulated paths is $100$. Consider an European call option with the expiration date $T$ and the strike price $K$, as 
\begin{equation*}
f(S_T)=\max(S_T-K,0),
\end{equation*}
and $K=S_0 \times u$ which $u=0.3, 0.45, 1.00, \ldots, 6.45, 7$. The specifications of the computer system with which the program is implemented are Intel(R) Core i$7-9700$K CPU and $64$ GB Memory.
\begin{Example}
Let $J_{t,z} =z$ for $z \in \mathbb{R}$ when the jump sizes follow a Gaussian distribution with $\mu_J=-0.10$, mean of the jump sizes, and the standard deviation of jump sizes $\sigma_J=0.50$ and density function $f_J$. Mean ($\mu_J$) determines the average size of the jumps and if $\mu_J > 0$, the jumps tend to be upward, if $\mu_J < 0$, the jumps tend to be downward. The variance $\sigma_J^2$ controls the spread of the jump sizes. A larger $\sigma_J^2$ results in more variability in the jump sizes. The probability distribution function (PDE) is symmetric around $\mu_J$ if $\sigma_J$ is fixed.\\
In these models, such as the Merton jump-diffusion model, the PDF $f_J(J)$ determines the likelihood of different jump sizes in the jump term ($e^J - 1$) of the process. These jumps add discontinuities to the asset price dynamics, making the process more realistic for financial modeling.\\
In Table 1, we present two types of errors, the mean square error(MSE) and the absolute(ABS) error, for the Euler scheme in pricing the underlying asset process in this model. Also, Table 2  presents the same errors for pricing the European call option in this scheme, based on Theorem \ref{thm1}. Here, $n$ is the number of time discretizations and the number of simulated paths is $100$.\\
\begin{table}[H]
	\centering
	\caption{Log2 error of asset price in Gaussian model}
	\scalebox{1}{
		\begin{tabular}{ccccc}
			\hline
			$n$ & \;\;$ABS\ Error$ & \;\;\;\;\;\;$MSE\ Error$\\
			\hline 
			$100$ &\;\;\;\;\; $-5.2651$ &\;\;\;\;\; $-9.2632$\\
			$200$ &\;\;\;\;\; $-6.0292$ &\;\;\;\;\; $-10.7205$\\
			$400$ &\;\;\;\;\; $ -6.6988$ &\;\;\;\;\; $-11.9393$\\
			$800$ &\;\;\;\;\; $ -7.3824$ &\;\;\;\;\; $-13.1895$\\
			$1600$ &\;\;\;\;\; $-7.9956$ &\;\;\;\;\; $-14.8794$\\
			$3200$ &\;\;\;\;\; $-8.4803$ &\;\;\;\;\; $-15.8664$\\
			$6400$ &\;\;\;\;\; $-8.9430$ &\;\;\;\;\; $-16.8042$\\
			$12800$ &\;\;\;\;\; $-9.5871$ &\;\;\;\;\; $-18.1038$\\
			$25600$ &\;\;\;\;\; $-10.1018$ &\;\;\;\;\; $-19.1524$\\
			\hline
	\end{tabular}}
\end{table}
\begin{table}[H]
	\centering
	\caption{Log2 error of European option price in Gaussian model}
	\scalebox{1}{
		\begin{tabular}{ccccc}
			\hline
			$n$ &\;\; $ABS\ Error$ &\;\;\;\;\;\; $MSE\ Error$\\
			\hline 
			$100$ &\;\;\;\;\; $-4.0973$ &\;\;\;\;\; $-7.7954$\\
			$200$ &\;\;\;\;\; $-5.5223$ &\;\;\;\;\; $-10.6305$\\
			$400$ &\;\;\;\;\; $-4.8835$ &\;\;\;\;\; $-9.4318$\\
			$800$ &\;\;\;\;\; $-5.4067$ &\;\;\;\;\; $-10.5862$\\
			$1600$ &\;\;\;\;\; $-7.9447$ &\;\;\;\;\; $-15.4318$\\
			$3200$ &\;\;\;\;\; $ -8.7695$ &\;\;\;\;\; $-17.1613$\\
			$6400$ &\;\;\;\;\; $-8.5679$ &\;\;\;\;\; $-16.7704$\\
			$12800$ &\;\;\;\;\; $-12.0192$ &\;\;\;\;\; $-23.3891$\\
			$25600$ &\;\;\;\;\; $-11.3544$ &\;\;\;\;\; $-22.2063$\\
			\hline
	\end{tabular}}
\end{table}
The numerical results in Figure \ref{figure2} show the convergence of the method for European option prices in the normal jumps. The figure is plotted at a scale $0.001$.\\
\begin{figure}[H]
	\centering
	\includegraphics[scale=0.30]{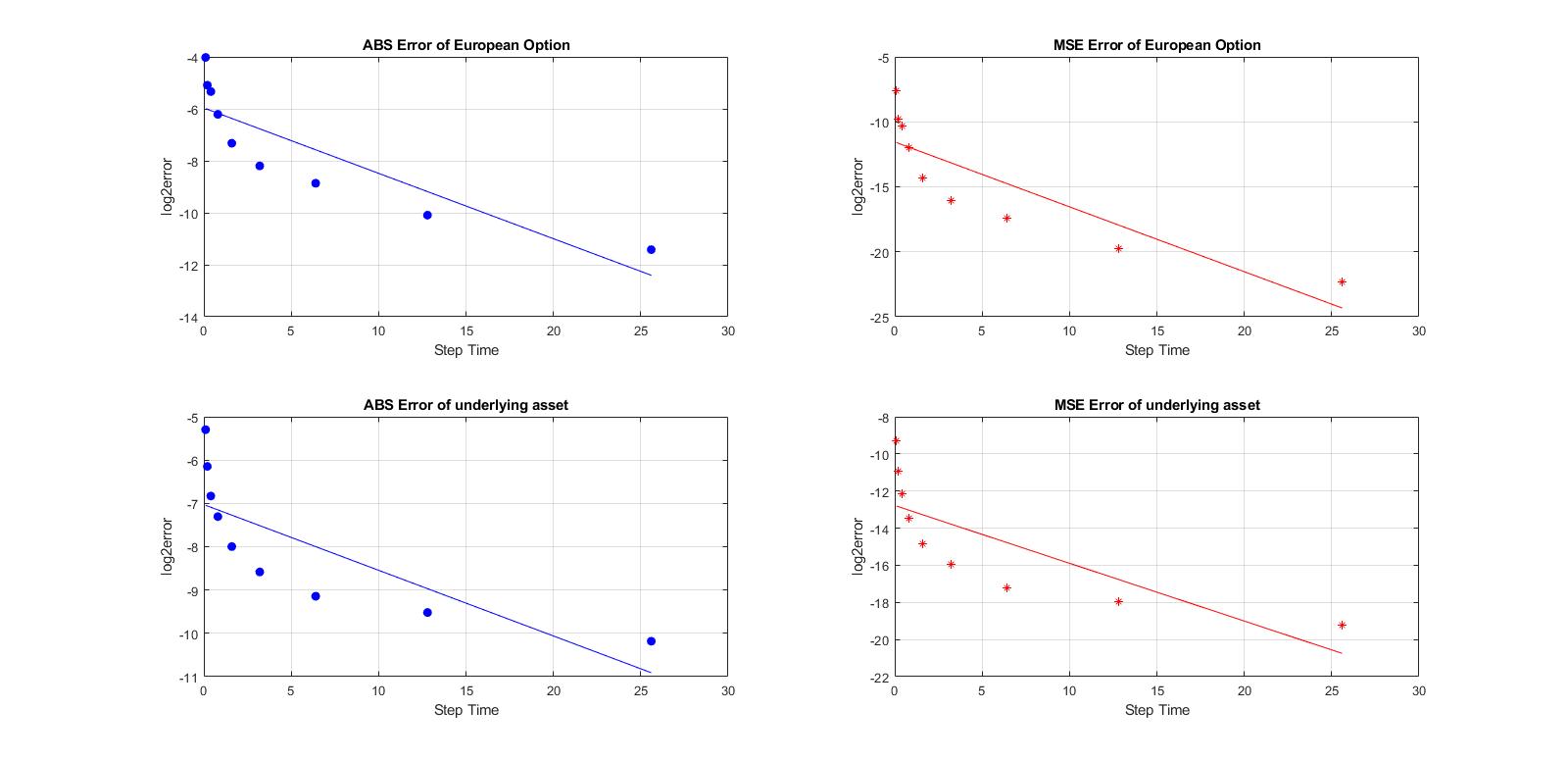}
	\caption{Error of asset price and European call option with normal jumps and 100 paths simulation.}\label{figure2}
\end{figure}
Figure \ref{figure1} shows the pricing of a European option with simulated paths is $1000$ and $K=S_0 \times u$ which $u=0.3, 0.45, 1.00, \ldots, 6.45, 7$.
\begin{figure}[H]
	\centering
	\includegraphics[scale=0.13]{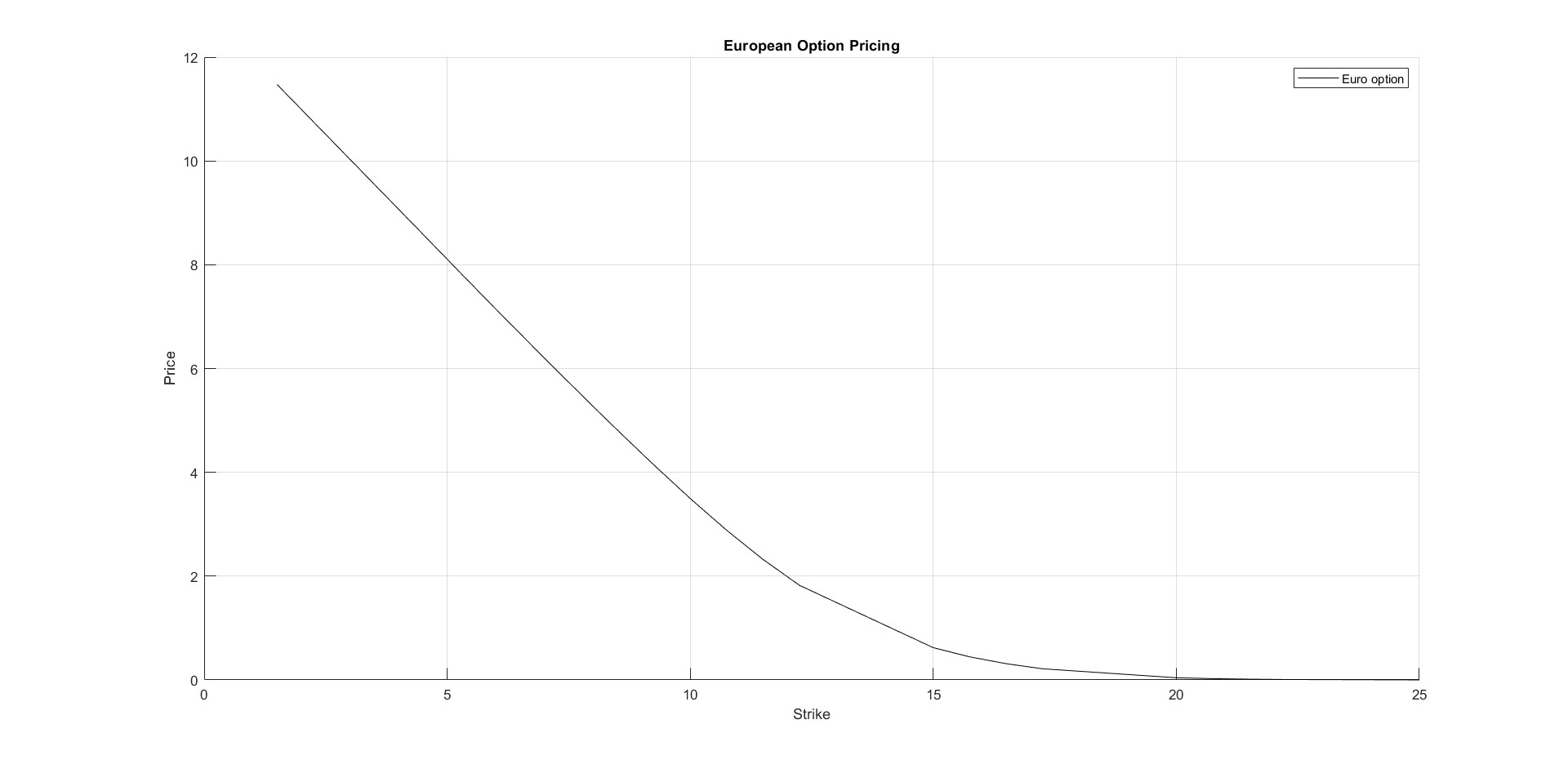}
	\begin{center}\caption{Pricing of European call option for $T=1, S_0=5$ and the function $J_{t,z}=z$ with Gaussian jump distribution and 1000 paths simulation.}\label{figure1}
\end{center}
\end{figure}
\end{Example}
\begin{Example}
In Kou moel, let $J_{t,z} =\lvert z\rvert$, for $z \in \mathbb{R}$, $\eta_1=10$, and $\eta_2=5$, $p_{\text{up}}=0.50$, up jump probability for selecting the exponential component. in Tables 3 and 4, we use the error types of Euler method in this model with $100$ paths simulation.
\begin{table}[H]
	\centering
	\caption{Log2 error of asset price in Kou model}
	\scalebox{1}{
		\begin{tabular}{ccccc}
			\hline
			$n$ & \;\;$ABS\ Error$ & \;\;\;\;\;\;$MSE\ Error$\\
			\hline 
			$100$ &\;\;\;\;\; $-5.2981$ &\;\;\;\;\; $-9.2857$\\
			$200$ &\;\;\;\;\; $-6.1458 $ &\;\;\;\;\; $-10.9481$\\
			$400$ &\;\;\;\;\; $-6.8295$ &\;\;\;\;\; $-12.1544$\\
			$800$ &\;\;\;\;\; $-7.3050$ &\;\;\;\;\; $-13.4421$\\
			$1600$ &\;\;\;\;\; $-7.9969$ &\;\;\;\;\; $-14.8213$\\
			$3200$ &\;\;\;\;\; $-8.5831$ &\;\;\;\;\; $-15.9298$\\
			$6400$ &\;\;\;\;\; $-9.1437$ &\;\;\;\;\; $-17.2007$\\
			$12800$ &\;\;\;\;\; $-9.5195$ &\;\;\;\;\; $-17.9369$\\
			$25600$ &\;\;\;\;\; $-10.1825$ &\;\;\;\;\; $-19.2163$\\
			\hline
	\end{tabular}}
\end{table}
\begin{table}[H]
	\centering
	\caption{Log2 error of pricing European call option in Kou model}
	\scalebox{1}{
		\begin{tabular}{ccccc}
			\hline
			$n$ &\;\; $ABS\ Error$ &\;\;\;\;\;\; $MSE\ Error$\\
			\hline 
			$100$ &\;\;\;\;\; $-4.0208$ &\;\;\;\;\; $-7.6127$\\
			$200$ &\;\;\;\;\; $-5.0903$ &\;\;\;\;\; $-9.7733$\\
			$400$ &\;\;\;\;\; $-5.3342$ &\;\;\;\;\; $-10.2887$\\
			$800$ &\;\;\;\;\; $-6.2112$ &\;\;\;\;\; $-11.9935$\\
			$1600$ &\;\;\;\;\; $-7.3233$ &\;\;\;\;\; $-14.3114$\\
			$3200$ &\;\;\;\;\; $-8.1958$ &\;\;\;\;\; $-16.0198$\\
			$6400$ &\;\;\;\;\; $-8.8629$ &\;\;\;\;\; $-17.3871$\\
			$12800$ &\;\;\;\;\; $-10.0910$ &\;\;\;\;\; $-19.7188$\\
			$25600$ &\;\;\;\;\; $-11.4178$ &\;\;\;\;\; $-22.3479$\\
			\hline
	\end{tabular}}
\end{table}
\begin{figure}[H]
	\centering
	\includegraphics[scale=0.30]{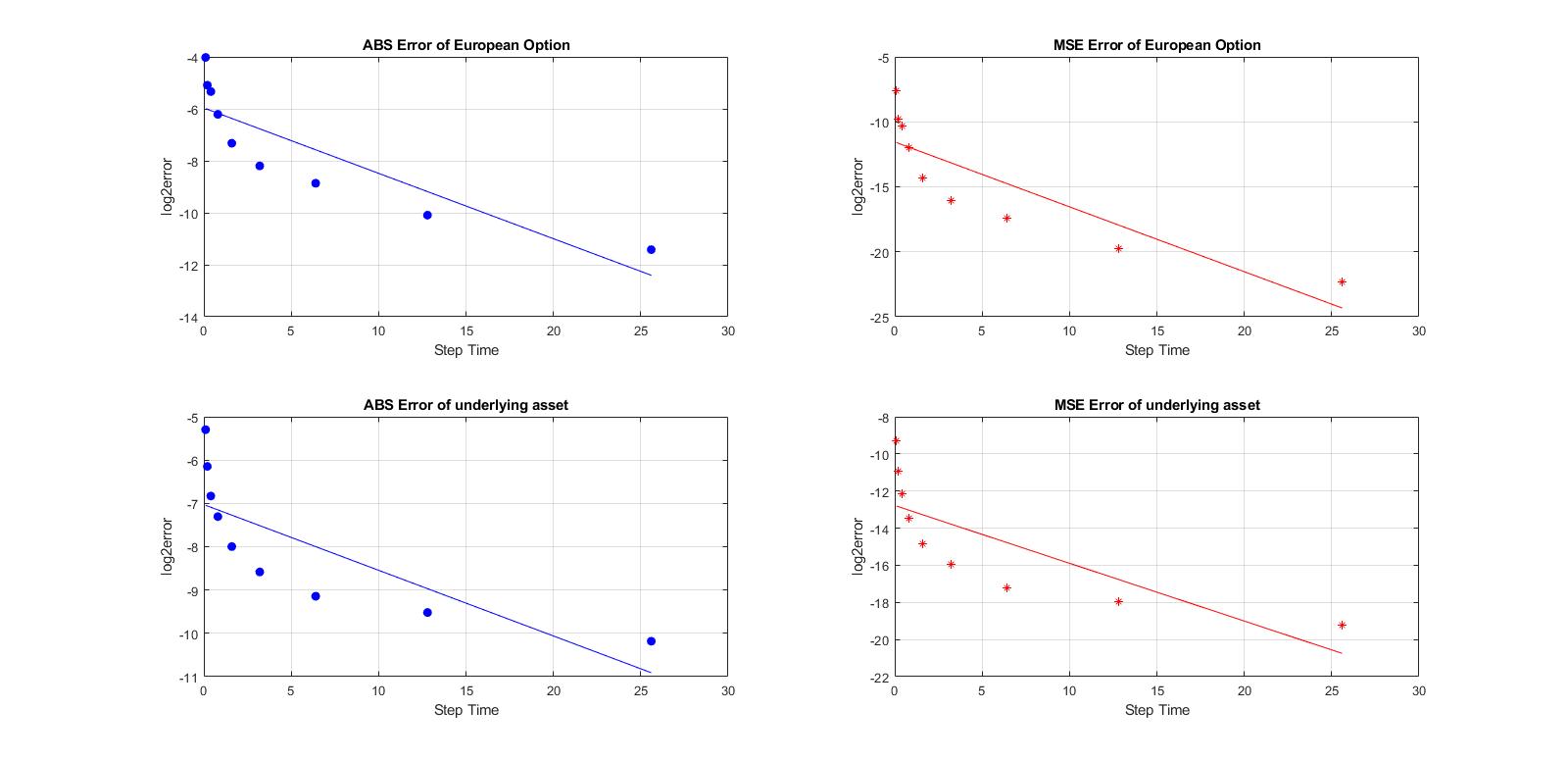}
	\caption{Error of price model and European option with Double Exponential Distribution jumps (Kou model) and 100 paths simulation.}\label{figure1'}
\end{figure}
The numerical results in Figure \ref{figure1'} show the convergence of the method for European option price in the Kou model. The figure is plotted at a scale $0.001$.\\
\end{Example}

\subsection{Delta in the first and second approach}
In this subsection, we calculate the delta in two approaches of computing the Malliavin derivative for the European call option and show the results.\\ 
The exact expression for $\Delta$ is
\begin{equation*}
\Delta=\mathbb{E}[H_K(S_T)\frac{S_T}{S_0}],
\end{equation*}
whereas the symmetric finite difference approach gives
\begin{equation*}
\Delta=\frac{\partial}{\partial S_0}\mathbb{E}[\max(S_T-K,0)]=\frac{F(S_0+h)-F(S_0-h)}{2h},
\end{equation*}
where $F(S_0) =\mathbb{E}[f(S_T)\vert S_0]$, and $h$ is an arbitrary small constant. \\
In figures \ref{figkappa} and \ref{figsigma2}, the sensitivity of the price of a call option are presented with respect to the parameters of the stochastic intensity model; $\kappa$ and $\sigma$.   
\begin{figure}[h!]
\begin{center}
 	\includegraphics[scale=0.18]{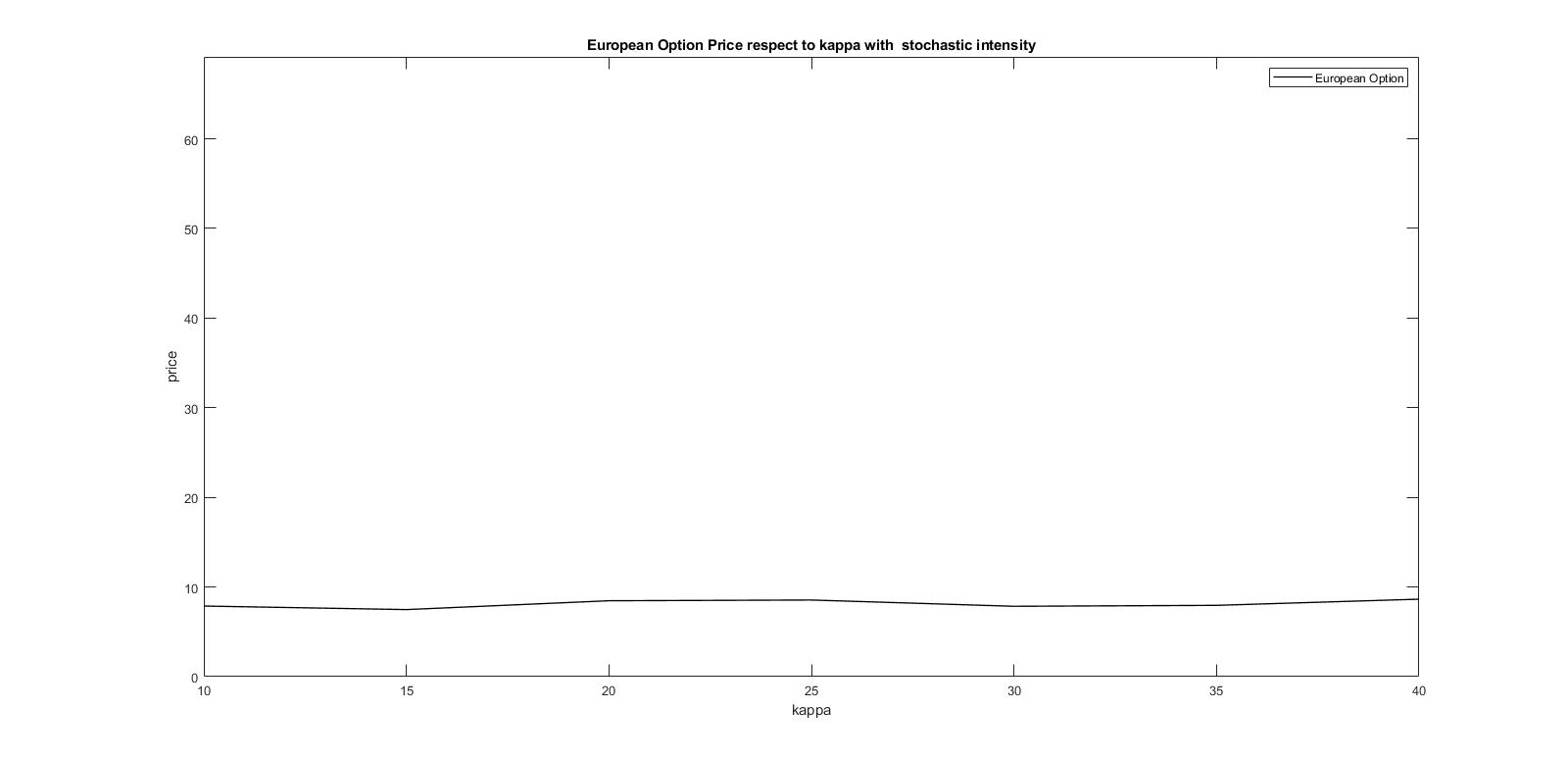}
  \caption{sensitivity of price with respect to k}\label{figkappa}
    \includegraphics[scale=0.18]{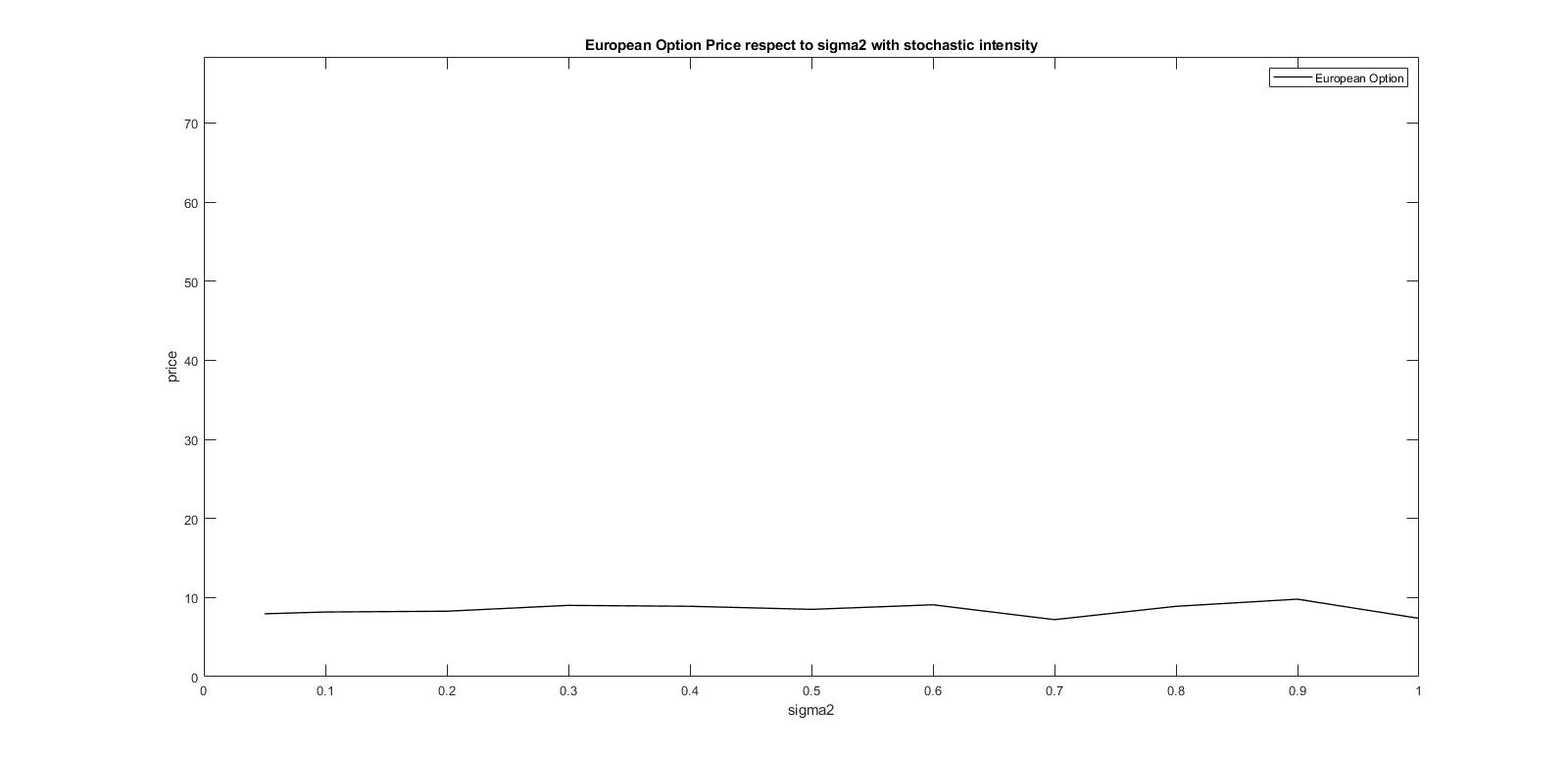}
    \caption{sensitivity of price with respect to sigma2}\label{figsigma2}
 \end{center}
   \label{figureparameter}
\end{figure}

Figures \ref{figure4} and \ref{figure3} show the behaviour of these four expressions $\Delta$, $\Delta^W$, $\Delta^N$ and $\Delta^W/2+\Delta^N/2$ for $\sigma_1=0.10$, $\sigma_2=0.05$, $\Theta=0.30$, $\kappa=0.5$, $\mu=0.01$, $T=1$, $\lambda_0=0.10$, $S_0=5$, $K=S_0\times1.2$ and time discretization $h=0.0001$. The jumps are generated by a normal distribution with a mean of $-0.10$ and a standard deviation of $0.50$ which satisfies Condition H1. 
The exact solution is $0.0481509$. The execution time of the program code in the Malliavin method and finite difference method in the first approach are $2.2104\times10^{4}$ and $4.4111\times10^{4}$, and in the second approach are $2.2817\times10^{4}$ and $4.5626\times10^{4}$ respectively. The error of four expressions is presented in Table \ref{tab5}.
\begin{figure}[H]
\centerline{
	\includegraphics[scale=0.16]{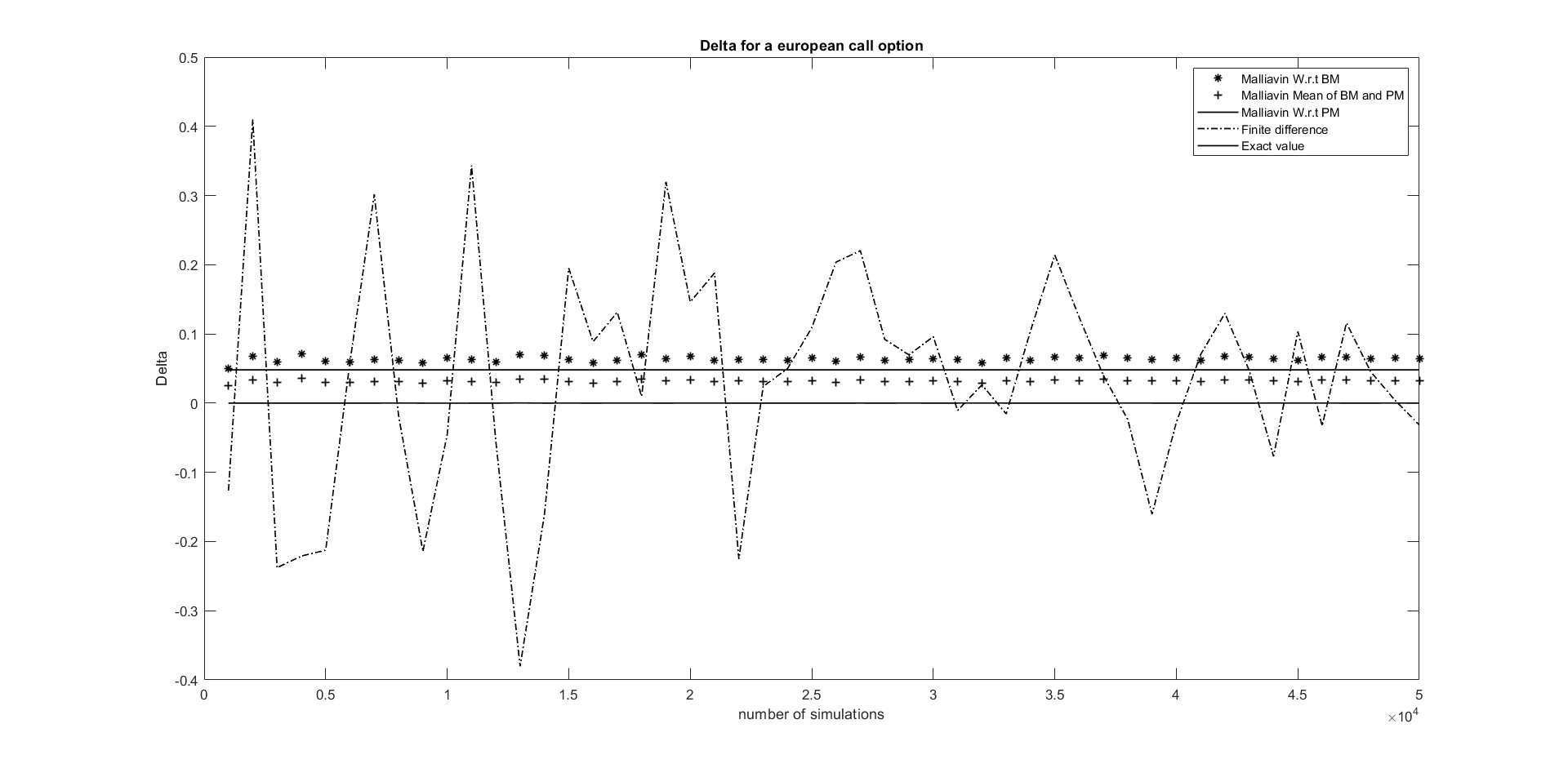}
	}
	\caption{Greek Delta for European call option in the first approach for $J_{t,z}=\vert z\vert$ and Gaussian jump distribution.}\label{figure4}
	\centerline{\includegraphics[scale=0.20]{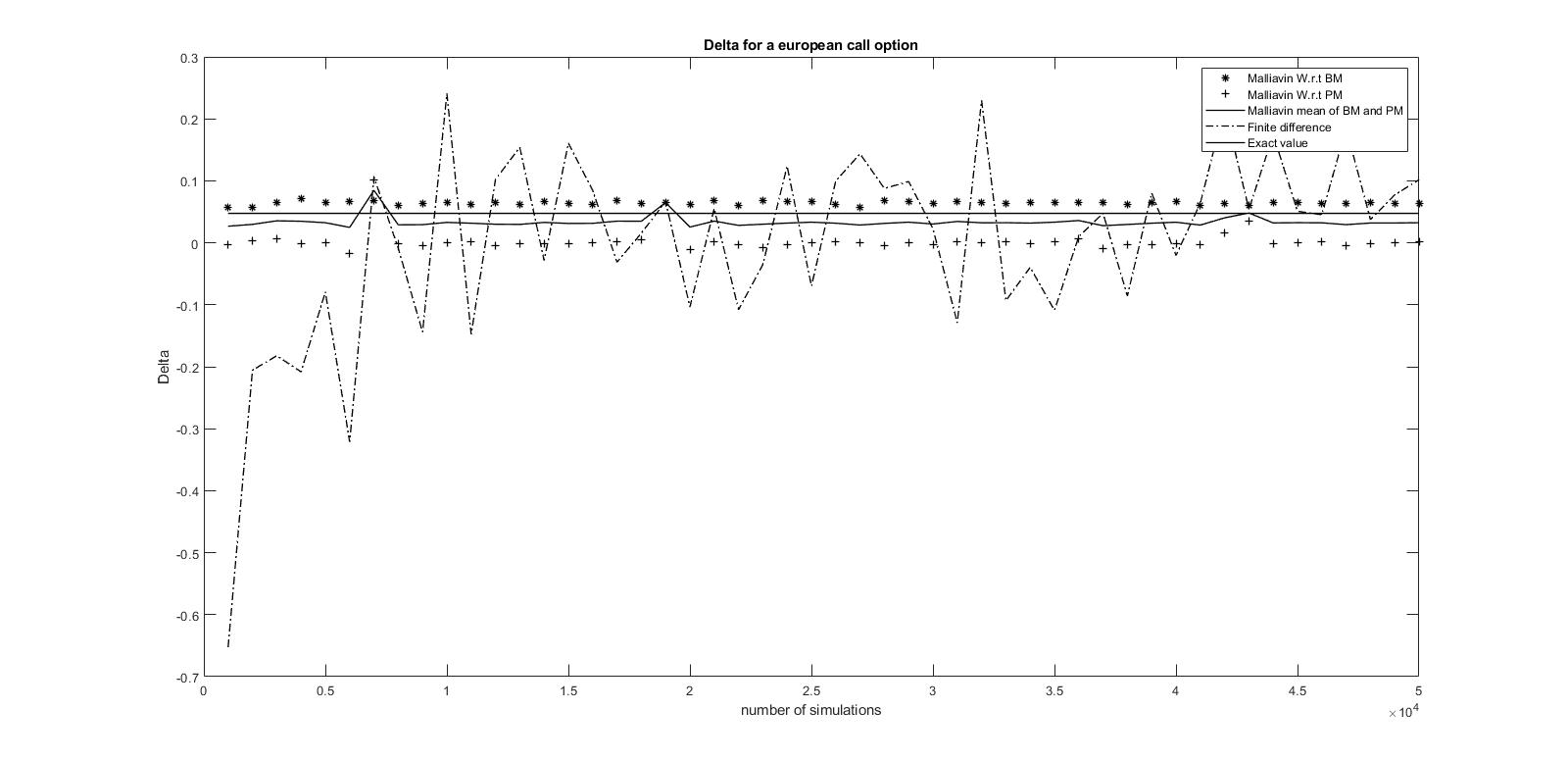}
	}
	\caption{Greek Delta for European call option in the second approach for $J_{t,z}=\vert z\vert$ and Gaussian jump distribution.}\label{figure3}
\end{figure}
\begin{table}[ht!]
\begin{center}
	\caption{The mean square error of four methods}
	\scalebox{1}{
		\begin{tabular}{ccccc}
			\hline
			$The\ Method$ & $MSE\ of\ the\ first\ approach$ & $MSE\ of\ the\ second\ approach$ \\
			\hline 
			$Winner-Malliavin\ Weight$ & $0.0004$ & $0.0003$\\
			$Poisson-Malliavin\ Weight$ & $0.0023$ & $0.0021$\\
			$Mean\ Winner\ and\ Poisson-Malliavin\ Weight $ & $0.0003$ & $0.0002$\\
			$Symmetric\ Finite\ Difference$ &  $0.0254$ & $0.0190$\\
			\hline
	\end{tabular}}\label{tab5}
\end{center}
\end{table}

\section{Funding}
This work is based upon research funded by Iran National Science Foundation(INSF) under project No.4022879.
\section{Conclusions}
The main purpose of this article is to study the pricing of financial derivatives and calculate their delta in a stochastic model with stochastic intensity by using the Mallivain calculus. In the presence of the Malliavin derivative of the intensity, specific Wiener directions are identified and employed in the duality formula of the Gaussian case, which is used to calculate the delta and the price of financial derivatives. This topic, particularly delta computation, is also addressed in Poisson space using two distinct approaches. We also prove the convergence of the Euler method to the true solution of asset price and also the price of a European call option. Finally, numerical results in two models, Gaussian jumps and the Kou model, are displayed and the delta approximation is presented to compare the price sensitivity computation in two methods- the finite difference method and the Malliavin method- against the exact solution in models with jumps and stochastic intensity for asset prices and financial derivatives. \\
The methodology developed in this paper can be extended to other pricing problems and Greeks associated with stochastic volatility processes and fractional Brownian motion. These extensions are left for future work.

\section{Acknowledgments}
We thank the referees for their valuable suggestions and comments that will improve the paper.

\section{Financial disclouser}
No data were created or analyzed in this study.

\section{Conflict of interest}
The authors declare no potential conflict of interest.


\end{document}